\documentclass[12pt]{article}
\usepackage{geometry}
\usepackage{calc}
\usepackage{mathtools}
\usepackage[ruled,vlined, noend, linesnumbered]{algorithm2e}
\usepackage{amsmath}
\usepackage{amsthm}
\usepackage{amssymb}
\usepackage{graphicx}
\usepackage{setspace}
\usepackage{url}
\providecommand{\tabularnewline}{\\}

\usepackage{verbatim}

\usepackage{subfigure}

\theoremstyle{remark}
\newtheorem{claim}{Claim}

\usepackage[font=scriptsize]{caption}

\newcommand{\bd}{\text{Beta}}

\newcommand{\tht}{\theta}

\newcommand{\omg}{\omega}

\newcommand{\C}{\text{$\mathbb C$}}

\newcommand{\N}{\text{$\mathbb N$}}

\newcommand{\R}{\text{$\mathbb{R}$}}

\newcommand{\omgg}{\boldsymbol{\omg}}

\newcommand{\pp}{\boldsymbol{p}}

\newcommand{\qq}{\boldsymbol{q}}

\newcommand{\uu}{\boldsymbol{u}}

\newcommand{\vv}{\boldsymbol{v}}
\newcommand{\ww}{\boldsymbol{w}}

\newcommand{\xx}{\boldsymbol{x}}

\newcommand{\yy}{\boldsymbol{y}}

\begin{document}

\title{Exactly computing the tail of the Poisson-Binomial Distribution}
\author{Noah Peres, Andrew Lee and Uri Keich}
\maketitle

\begin{abstract}
We offer ShiftConvolvePoibin, a fast exact method to compute the tail of a Poisson-Binomial distribution (PBD).
Our method employs an exponential shift to retain its accuracy when computing a tail probability, and in practice
we find that it is immune to the significant relative errors that other methods, exact or approximate,
can suffer from when computing very small tail probabilities of the PBD.
The accompanying R package is also competitive with the fastest
implementations for computing the entire PBD. 
\end{abstract}

\section{Introduction}

Let $X_{i}$ for $i=1,\dots,N$ be independent Bernoulli random variables
(RVs) with corresponding probabilities of success $p_{i}$. The distribution
of $X=\sum_{1}^{N}X_{i}$ is called a Poisson-binomial distribution
(PBD) and it clearly generalizes the binomial distribution ($p_{i}\equiv p$
for all $i$).

Biscarri et al.~highlighted the considerable interest in using the
PBD in diverse areas of scientific research ranging from genetics
through survey sampling to sports analysis~\cite{Biscarri2018}.
Coinciding with this scientific interest was a continuing effort on
part of the statistical community to develop efficient and accurate
tools for evaluating the significance of tests based on the PBD. Most
of these tools, which range from approximation-based ones, including
Poisson and normal approximations, to a growing recent interest in
exact methods are also reviewed in~\cite{Biscarri2018}.

In this paper we focus on computing the tail probability of the PBD. Specifically
we look at evaluating $P\left(X\ge x\right)$ for a Poisson-binomial (PB) RV $X$
focusing on accurately computing this (right) tail probability when it is small.
Note that a left tail can be readily transformed into a right tail by considering
$X'=N-X=\sum_{1}^{N}\left(1-X_{i}\right)$ which is again a PB RV.

Recently Madsen et al.~pointed out that when using Hong's popular
DFT-CF~\cite{Hong2013} to compute the right tail of the PBD one
can expect extremely high \emph{relative errors} when the actual value
is smaller than about $10^{-16}$~\cite{Madsen2017}.\footnote{If the true value $s\ne0$ is computed as $\tilde{s}$ then the associated
\emph{relative error} is defined as $\left|\left(\tilde{s}-s\right)/s\right|$.} While many statisticians would not consider that a problem, and indeed
it is not a problem in the canonical 5\% significance cutoff scenario,
Madsen et al.'s motivation comes from bioinformatics research where
one often needs to accurately evaluate even much smaller tails. In
such cases relying on DFT-CF could lead to significant errors in the
downstream analysis.

As an alternative, Madsen et al.~develop a saddle-point based approximation
method which, as they demonstrate, has a much better-behaved relative
error when computing the p-value of large values of the statistic
$X$~\cite{Madsen2017}. Moreover, as Madsen et al.~argue, the
runtime complexity of their approach is $O\left(N\right)$ compared
with DFT-CF's complexity of $O\left(N^{2}\right)$ making their method
a seemingly win-win proposition for someone interested in evaluating
the right tail of the PBD.

In this paper we show that while the saddlepoint approximation method
of Madsen et al.~(SA) is generally quite accurate, there are cases
where it also suffers from significantly large errors. Particularly
problematic are cases where the actual p-value is close to 1 but SA
reports much smaller values (e.g., Figure \ref{fig:SP-problem} below).
Madsen et al.~acknowledged that SA ``is not suited for calculating
large (not significant) p-values ($>0.1$)'', however how would the
user know the p-value is $>0.1$ if SA reports a value that
is much closer to 0 (e.g., consider $s_{0}<4000$ in panel A of Figure \ref{fig:SP-problem}
below)?

Going back to the extremely large relative errors that DFT-CF can
induce, we show below that the same applies to DC-FFT, which is the
more recent method of Biscarri et al.~\cite{Biscarri2018}. Note
that both DFT-CF and DC-FFT are exact methods, that is, they compute
the probability mass function (pmf) of the PBD using the underlying
distribution rather than relying on an asymptotic approximation. So
how can the approximation-based SA be much more accurate than those
exact methods?!

Before explaining this we would like to clarify that both DFT-CF
and DC-FFT are accurate as long as you gauge their accuracy using
the total absolute error (TAE) as your figure of merit. Indeed, as
reported by both Hong and Biscarri et al., in that case these exact
methods live up to their name with the error rarely exceeding $10^{-10}$.
However, as noted above the relative error can tell a very different
story.

As explained, for example in \cite{Wilson2016b}, the source of the
extremely large relative errors is the Discrete Fourier Transform
(DFT, defined below), which both DFT-CF and DC-FFT rely on. Specifically,
the DFT involves summing many positive and negative numbers so the accumulated roundoff errors
are amplified by intermediate cancellations and potentially create extremely large relative errors
in the final result --- something that we observe with DFT-CF and DC-FFT.

In contrast, if we add up only non-negative numbers then the relative
error is well-controlled. In particular, the exact method of Direct
Convolution (DC) --- proposed by Biscarri et al. (Algorithm 1~\cite{Biscarri2018}),
as well as by Madsen et al., can be considered as the gold standard
in terms of accuracy. The downside of DC is that its runtime complexity
is $O\left(N^{2}\right)$ making it potentially forbiddingly slow
for large values of $N$. Indeed, DC-FFT was specifically designed
to be much faster than DC.

In this paper we introduce ShiftConvolvePoibin (or ShiftConvolve for short), a novel exact method with
a close-to-linear runtime complexity of $O\left(N\left(\log N\right)^{2}\right)$
and which, in practice, is on par with DC in terms of its control
of the relative error. ShiftConvolve accomplishes this by relying
on the same exponential shift idea that was originally employed in~\cite{Keich2005}
to control the numerical errors introduced by the DFT when computing
a certain tail probability (see also \cite{Wilson2016b}).

\section{Addressing the accuracy problem}

\subsection{DFT and the relative error problem}

The DFT operator, $D$, and its inverse, $D^{-1}$, are linear operators
defined on $\C^{n}$ as

\begin{equation}
\begin{aligned}(D\xx)(k) & \coloneqq\sum_{j=0}^{n-1}\xx(j)e^{-i2\pi kj/n}\qquad k=0,\dots,n-1\\
(D^{-1}\yy)(k) & \coloneqq\frac{1}{n}\sum_{j=0}^{n-1}\yy(j)e^{i2\pi kj/n}\qquad k=0,\dots,n-1.
\end{aligned}
\label{eq:DFT}
\end{equation}
It is clear from their definitions that computing the real and imaginary
components of the result generally involves adding up both positive
and negative numbers creating the significant relative errors we alluded
to.

In practice the DFT is almost invariably calculated by the Fast Fourier
Transform (FFT) which has a time complexity of $O\left(n\log n\right)$~\cite{Cooley:1965}.
Brisebarre et al.~provide an exhaustive overview of the analysis
of the error introduced by the FFT~\cite{Brisebarre2019} and it
should be stressed that in general a naive implementation of the DFT
would not be any more accurate than the FFT~\cite{Schatzman1996}.
Regardless, our goal here is not to bound the error, rather we follow
up on Madsen et al.~in pointing out the potentially catastrophic
effect the FFT can have on the relative error and to offer a solution
to this problem.

Both DFT-CF and DC-FFT rely on the DFT (and its inverse): the first
to invert the characteristic function of the PBD and the second to
perform its convolutions as explained below. In both cases the DFT
can severely compromise the relative accuracy of the computed values
as shown in panel A of Figure \ref{fig:DFT_error}: while the DFT-relying
methods coincide with the accurate DC for pmf entires that are larger
than $\approx10^{-16}$, values that are smaller cannot be recovered
by DFT-CF and DC-FFT (DFT-CF often returns 0 in this case as is evident
by the gaps in the green logarithmic curve). In particular, using
DFT-CF or DC-FFT to compute a right tail probability for large values
of $x$ will typically yield a result that is orders of magnitude off in terms of the
relative accuracy: the correct values as per DC are orders of magnitude
smaller than the values reported by DC-FFT, whereas for DFT-CF the
reported values vary between 0 (100\% relative error) and the same
order as DC-FFT's reported values.\footnote{Note that the absolute error is still small: no more than $\approx10^{-16}$
but our interest here is in accurately recovering the small values.}

Panel A of Figure \ref{fig:DFT_error} suggests that it is impossible
to recover the smaller entries of the pmf because of the numerical
errors inherent to the DFT. There is however a solution that was first
suggested in~\cite{Keich2005} and that uses an exponential shift
to overcome the numerical errors.

\begin{figure}
\centering %
\begin{tabular}{ll}
A. Original pmf, no shift & B. Original pmf as well as shifted (DC-FFT)\tabularnewline
\includegraphics[width=3in]{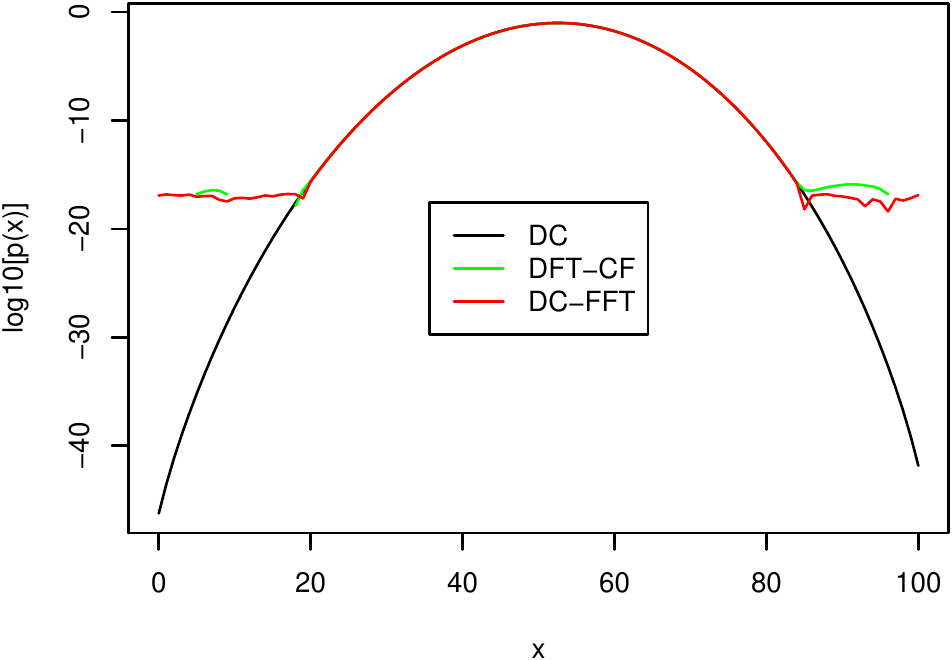} & \includegraphics[width=3in]{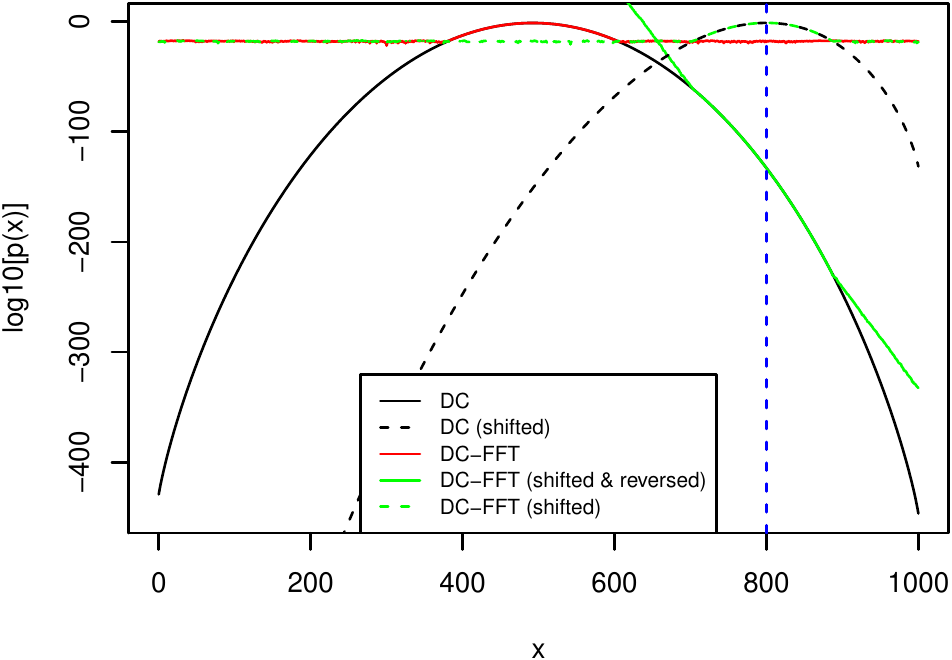}\tabularnewline
C. Relative error (DC-FFT / shifted) & D. Original pmf (DFT-CF / shifted)\tabularnewline
\includegraphics[width=3in]{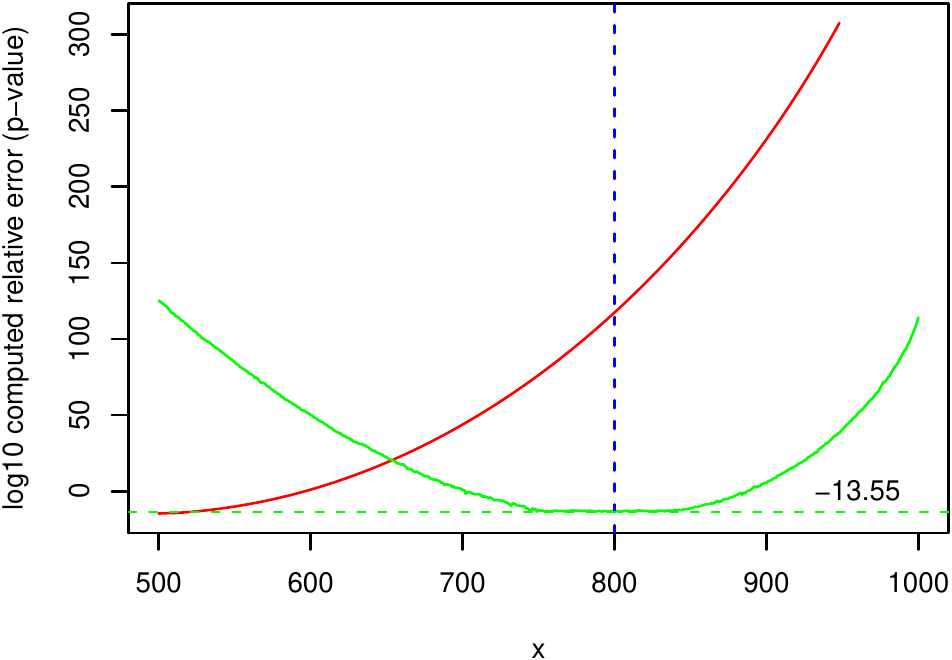} & \includegraphics[width=3in]{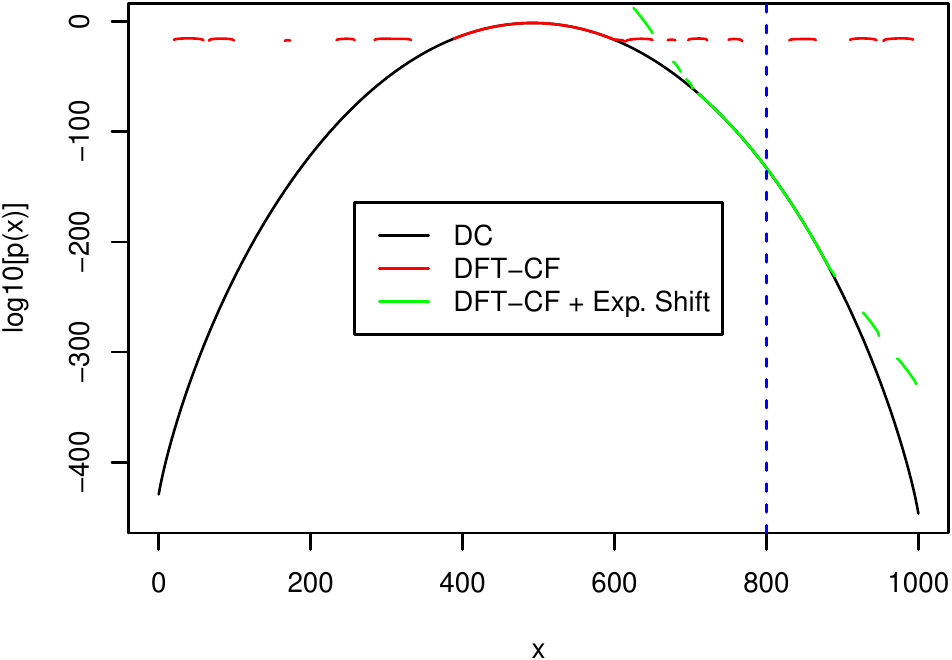}\tabularnewline
\end{tabular}\caption{\textbf{The accuracy problem. }(A) The (log of the) PB pmf computed by
the accurate DC is contrasted with the values computed by DFT-CF and
DC-FFT both of which rely on the DFT: values below $\approx10^{-16}$
cannot be recovered. Missing values of DFT-CF correspond to reported
0 values. The PBD here is defined using $N=100$ values of $p_{i}$ that were independently and uniformly
sampled. (B) The log of the PB pmf computed by DC-FFT compared with a shifted variant.
(C) The (log of the) relative error in computing the right tail using DC-FFT and its
shifted variant. (D)
Same as panel B but using DFT-CF and a shifted variant of DFT-CF.
Gaps correspond to reported 0 values. \label{fig:DFT_error}}
\end{figure}

\subsection{The exponential shift}

Let $\qq$ be a PB pmf supported on $0,1,\dots,n$ then the
exponentially shifted version of $\qq$ is defined as
\begin{equation}
\qq_{\theta}(k)=\qq(k)e^{k\theta}/M_{\qq}(\theta)\qquad k=0,1,\dots,n,\label{eq:shifted_pmf}
\end{equation}
where $M_{\qq}(\theta)$ is the moment generating function (MGF) of
$\qq$. Note that dividing by the MGF guarantees that $\qq_{\tht}$
is a proper pmf.

The basic idea behind the introduction of the exponential shift is
that the FFT-based methods are still accurate for the larger values
so all we need to do is make sure the values we are actually interested
in are ``large enough''.

The approach is visualized in panel B of
Figure \ref{fig:DFT_error}, which focuses on accurately recovering
$P\left(X\ge800\right)$. As seen by the red curves in panels B and
C of the figure, DC-FFT does not allow us to accurately gauge this
probability which is many order of magnitude smaller than the values
reported by DC-FFT.

However, when the proper exponential shift is applied to
the pmf (panel B, dashed black curve) the values that we are interested
in are sufficiently inflated so that DC-FFT accurately recovers the
section of the pmf about $x=800$ (panel B, dashed green curve).

Finally, reversing the shift we note that the combination
of DC-FFT with the shift and its reversal allows us to accurately
recover the tail probability for any $x$ in the neighborhood of 800:
the relative error is miniscule ($\approx10^{-13}$, panels B and
C, solid green curves). Panel D suggests that the same principle would
work for DFT-CF.

The reason we use an exponential shift rather than
some other arbitrary way to inflate the values we are interested in
is that the exponential shift can be readily applied,
as well as peeled off, or reversed in our context. Specifically:
\begin{claim}
The exponential shift commutes with the convolution operation ($\ast$),
that is, if $\pp$ and $\qq$ are pmfs defined on $0,1,\dots,n$ then
$\left(\pp\ast\qq\right)_{\tht}\equiv\pp_{\tht}\ast\qq_{\tht}$.
\end{claim}
\begin{proof}
Recall that if $X$ and $Y$ are independent $\N$-valued RVs with corresponding pmfs $\pp$ and $\qq$ then
\[
M_{\pp\ast\qq} \equiv M_{X+Y} \equiv M_X \cdot M_Y \equiv M_{\pp} \cdot M_{\qq} .
\]	
With this in mind for $k=0,\dots,2n$ we have
\[
\begin{aligned}\left(\pp\ast\qq\right)_{\tht}\left(k\right) & =\left(\pp\ast\qq\right)\left(k\right)\cdot e^{k\tht}/M_{\pp\ast\qq}\left(\tht\right)\\
 & =\sum_{i=0}^{k}\left(\pp(i)\cdot e^{i\tht}\right)\left(\qq\left(k-i\right)\cdot e^{\left(k-i\right)\tht}\right)/\left(M_{\pp}(\theta)\cdot M_{\qq}(\theta)\right)\\
 & =\left(\pp_{\tht}\ast\qq_{\tht}\right)\left(k\right).
\end{aligned}
\]
\end{proof}
It follows that the same holds for convolutions of any order and hence
that we can compute $\qq_{\tht}$, the $\tht$-shifted version of
the PB pmf $\qq$, by convolving the $\tht$-shifted versions of each
of the Bernoulli pmfs. The latter of course can be trivially
found because a $\tht$-shifted Bernoulli($p$) pmf is again a Bernoulli
pmf only with probability of success $p_{\tht}=pe^{\tht}/\left(1-p+pe^{\tht}\right)$.
Note that, in particular, an exponentially shifted PBD is also a PBD.

It follows that we can use DFT-CF and DC-FFT to compute the exponentially
shifted pmf $\pp_{\tht}$ of the PBD by applying them to $\tht$-shifted
Bernoullis. Of course, recovering $\qq$ from $\qq_{\tht}$ is a trivial
exercise of inverting \eqref{eq:shifted_pmf}, or equivalently applying
a shift of $-\tht$ to $\qq_{\tht}$. This is indeed the procedure
we applied when obtaining the green curves of panels B and D of Figure
\ref{fig:DFT_error}.

What is missing at this point is the protocol for determining $\tht$
given $x$. Here again we follow \cite{Keich2005}, and we define the
shift as the value of $\tht$ that minimizes the expression $\log M_{\qq}\left(\tht\right)-\tht\cdot x$.
The critical point of the latter function is the value $\tht$ for
which $M_{\qq}'\left(\tht\right)/M_{\qq}\left(\tht\right)=x$ but
$M_{\qq}'\left(\tht\right)/M_{\qq}\left(\tht\right)$ is simply the
expected value of the shifted $\qq_{\tht}$. Hence, finding $\tht$
amounts to solving
\begin{equation}
x=E\left(X_{\tht}\right)=\sum_{i=1}^{N}E\left(X_{i,\tht}\right)=\sum_{i=1}^{N}\frac{p_{i}e^{\theta}}{1-p_{i}+p_{i}e^{\theta}},\label{eq:tht}
\end{equation}
 where $X_{\tht}=\sum_{1}^{N}X_{i,\tht}$ is a $\qq_{\tht}$ distributed
RV and $X_{i,\tht}$ are independent Bernoulli($\pp_{\tht}(i)=p_ie^{\tht}/\left(1-p_i+p_ie^{\tht}\right)$)
RVs. In practice we solve \eqref{eq:tht} numerically using the \texttt{uniroot}
function in R.

\section{ShiftConvolvePoibin}

Combining our exponential shift protocol with either DFT-CF or DC-FFT
produces an exact method that is able to practically recover any right
tail probability (p-value) with a negligible relative error. The combined
general procedure is outlined in Algorithm~\ref{alg:shift-compute-PB-pmf-unshift}.

\begin{algorithm}
	\SetAlgoLined
	\DontPrintSemicolon
	\KwIn{$\pp = (p_0, p_1, ... ,p_{N-1})$ vector of success probabilities, observed value $s_0$}
	\tcc*{Lines 1-5: apply the exponential shift to $\pp$}
	\textbf{Function: }$\mu(\theta, \pp)$ $\longleftarrow $ sum$\left(\left[p_i\exp(\theta)/(1-p_i + p_i\exp(\theta)) \textbf{ for } p_i \in \pp\right]\right)$\;
	$\theta^* \longleftarrow$ find root $\theta$ such that $\mu(\theta, \pp) - s_0 = 0$\;
	$\pp_\tht \longleftarrow	$ empty list\;
	\For{$p_i \in \pp$}{
		$\pp_\tht(i)^+ \longleftarrow p_i \exp(\theta^*)/(1 - p_i + p_i\exp(\theta^*))$\;}
	$\qq_\tht\longleftarrow$ apply DC-FFT or DFT-CF to $\pp_\tht$\;\tcc*{Lines 7-10: reverse the exponential shift}
	$\qq \longleftarrow$ empty list\;
	$M \longleftarrow$ prod([$(1-p_i + p_i \exp(\theta^*))$ \textbf{for} $p_i \in \pp$])\;
	\For{$j \in 0,1,\dots,N-1$}{
		$\qq(j) \longleftarrow \qq_\tht(j) \exp(-j\cdot\theta)\cdot M$\;}
	\KwOut{$\qq$: the convolved pmf}
\caption{shift-compute-PB-pmf-unshift}
\label{alg:shift-compute-PB-pmf-unshift}
\end{algorithm}

Consider the variant of Algorithm~\ref{alg:shift-compute-PB-pmf-unshift}
where the pmf of the (shifted) PBD is computed with DC-FFT with its default
setting of $M=N$. In this case, starting with the $N$ (shifted) Bernoulli
pmfs at each step DC-FFT pairs all intermediate (shifted) PB pmfs
and convolves each pair while passing the resulting PB pmf to the
next step until at the last step it ends up with a single PB pmf
(Algorithm~\eqref{alg:DC-FFT-M-N}).

Each such pairwise convolution
is executed using the FFT based on the following identity. Suppose
that the pmfs $\pp$ and $\qq$ are supported on $\{0,1,\dotsc,m\}$
and $\{0,\dotsc,n\}$ respectively and for $Q\ge m+n+1$ embed $\pp$
and $\qq$ in $\R^{Q}$ by extending them with $Q-m$, respectively
$Q-n$, zeros. Then using the $Q$-dimensional DFT operator and its
inverse we have~\cite{Stockham1966}
\begin{equation}
\pp\ast\qq=D^{-1}(D\pp\odot D\qq),\label{eq:DFT-conv}
\end{equation}
where for $Q$-dimensional vectors $\uu$ and $\vv$, $(\uu\odot\vv)(i)=\uu(i)\vv(i)$.

\begin{algorithm}
	\SetAlgoLined
	\DontPrintSemicolon
	\KwIn{$p = (p_0, p_1, ... ,p_{N-1})$ vector of success probabilities}
	$V \longleftarrow$ empty list\;
	\For{$p_i \in p$}{
		add FFT($(1-p_i, p_i,0,0)$) to list $V$\;
	}
	$L$ $\longleftarrow$ length($V$)\;
	\While{$L > 2$}{
		$V^* \longleftarrow$ empty list\;
		\If{$L$ \textnormal{is odd}}{
			$n \longleftarrow$ length($\vv$ in $V$)\;
			add the $n-$dimensional vector $(1,1,...,1)$ to list $V$\;
		}
		split $V$ into pairs $(\vv_i, \vv_j)$\;
		\For{\textnormal{each pair }$(\vv_i, \vv_j)$ in $V$}{
			$\uu \longleftarrow$ FFTInverse(PointwiseMultiply($\vv_i, \vv_j$))\;
			$\uu \longleftarrow$ pad $\uu$ with length($\uu$) 0s\;
			$\vv^* \longleftarrow$ FFT($\uu$)\;
			add $\vv^*$ to list $V^*$\;
		}
		$V \longleftarrow V^*$\;
		$L \longleftarrow$ length($V$)\;
	}
	$\uu \longleftarrow$ FFTInverse(PointwiseMultiply($V[0],V[1]$))\;
	\KwOut{$\uu$}
\caption{Pair-aggregated-FFT-convolution (DC-FFT with $M=N$)}
\label{alg:DC-FFT-M-N}
\end{algorithm}

At first glance DC-FFT ($M=N$) seems wasteful because each intermediate
pmf is computed by applying $D^{-1}$ while at the next step $D$
is applied to the same pmf (lines 13 and 15). Of course, these two operators have different
dimensions (the latter's twice the former's) so this is not as wasteful
as it might initially look. Still, we next show how some work can
be saved by making a more efficient use of what has already been computed
in the previous step.

For $\uu\in\C^{n}$ let $\uu^{*}\in\C^{2n}$ be the zero-padded $2n$-dimensional
version of $\uu$: $\uu^{*}(i)=\uu(i)$ for $i=0,\dots,n-1$ and $\uu^{*}(i)=0$
for $i=n,\dots,2n-1$. Then for $k=0,\dots,n-1$
\begin{equation}
(D_{2n}\uu^{*})(2k)=\sum_{j=0}^{2n-1}\uu^{*}(j)e^{-i2\pi\left(2k\right)j/\left(2n\right)}=\sum_{j=0}^{n-1}\uu(j)e^{-i2\pi kj/n}=(D_{n}\uu)(k).\label{eq:even_2n}
\end{equation}
 That is, the even entries of $D_{2n}\uu^{*}$ coincide with $D_{n}\uu$
and therefore we do not need to recompute them.

As for the odd entries of $D_{2n}\uu^{*}$, for $k=0,\dots,n-1$
\begin{equation}
(D_{2n}\uu^{*})(2k+1)=\sum_{j=0}^{2n-1}\uu^{*}(j)e^{-i2\pi\left(2k+1\right)j/\left(2n\right)}=\sum_{j=0}^{n-1}\left(\uu(j)e^{-i\pi j/n}\right)e^{-i2\pi kj/n}=(D_{n}\left(\uu\odot\omgg\right))(k),\label{eq:odd_2n}
\end{equation}
where $\omgg(j)=e^{-i\pi j/n}$ for $j=0,\dots,n-1$ and $\odot$ again
is the coordinate-wise product. So while these odd entries do not
come for free they can be computed using an $n$-dimensional DFT rather
than a $2n$-dimensional one.

Algorithm~\ref{alg:more-efficient} takes advantage of the last two identities to speed up DC-FFT ($M=N$),
or Algorithm~\eqref{alg:DC-FFT-M-N}, by about 50\% for large $N$
(more on that in Section \ref{subsec:Complexity} below).

\begin{algorithm}
	\SetAlgoLined
	\DontPrintSemicolon
	\KwIn{$\pp = (p_0, p_1, ... ,p_{N-1})$ vector of success probabilities}
	$V \longleftarrow$ empty list\;
	\For{$p_i \in \pp$}{
		add FFT($(1-p_i, p_i,0,0)$) to list $V$\;
	}
	$L$ $\longleftarrow$ length($V$)\;
	\While{$L > 2$}{
		$V^* \longleftarrow$ empty list\;
		$n \longleftarrow$ length($\vv$ in $V$)\;
		$\omgg \longleftarrow$ [exp$(ij\pi/n)$ \textbf{for} $j \in$ [$0,1,...,n-1$]\;
		\If{$L$ \textnormal{is odd}}{
			add the $n$-dimensional vector $(1,1,...,1)$ to list $V$\;
		}
		split $V$ into pairs $(\vv_i, \vv_j)$\;
		\For{\textnormal{each pair }$(\vv_i, \vv_j)$ in $V$}{
				$\uu \longleftarrow$ FFTInverse(PointwiseMultiply($\vv_i, \vv_j$))\;
				$\ww \longleftarrow$ FFT(PointwiseMultiply($\uu, \omgg$))\;
				\tcc*{assign the components of $\vv$ and $\ww$ to even and odd components of $\vv^*$ respectively:}
				$\vv^* \longleftarrow$ ($\vv[0],\ww[0],\vv[1],\ww[1],...,\vv[n-1],\ww[n-1]$)\;
				add $\vv^*$ to list $V^*$\;
			}
			$V \longleftarrow V^*$\;
			$L \longleftarrow$ length($V$)\;
		}
	$\uu \longleftarrow$ FFTInverse(PointwiseMultiply($V[0],V[1]$))\;
	\KwOut{$\uu$}
\caption{Frugal-pair-aggregated-FFT-convolution (FPA-FFTC)}
\label{alg:more-efficient}
\end{algorithm}

Our ShiftConvolvePoibin algorithm (Algorithm \ref{alg:ShiftConvolve}) combines
Algorithm~\ref{alg:shift-compute-PB-pmf-unshift} with the latter,
more efficient, Algorithm~\ref{alg:more-efficient}.
Note that due to its built-in exponential shift ShiftConvolve can return the accurate logarithm
of the right tail probability even when the actual number can create an underflow or 0 when not using logs.

We have two implementations of ShiftConvolve where in both cases the critical part of the code is written in C and is wrapped in an R package.
The two versions differ in which code they use to execute the FFT: one version relies on FFTW~\cite{FFTW05} and requires
the user to install the FFTW package whereas the other uses minFFT~\cite{Mukhin2019} and is self-contained.
In both cases ShiftConvolve saves some runtime by taking advantage of the fact that the FFTInverse operation on line 13
of Algorithm~\ref{alg:more-efficient} should produce a real-valued vector.

\begin{algorithm}
	\SetAlgoLined
	\DontPrintSemicolon
	\KwIn{$\pp = (p_0, p_1, ... ,p_{N-1})$ vector of success probabilities, observed value $s_0$}
	Apply steps 1-5 of Algorithm \ref{alg:shift-compute-PB-pmf-unshift} to get the shifted $\pp_\tht$\;
	$\qq_\tht\longleftarrow$ FPA-FFTC($\pp_\tht$) \tcc*{apply Algorithm \ref{alg:more-efficient} to $\pp_\tht$}
	Apply steps 7-10 of Algorithm \ref{alg:shift-compute-PB-pmf-unshift} to get the unshifted $\qq$\;
	\KwOut{$\qq$: the convolved pmf}
\caption{ShiftConvolvePoibin}
	\label{alg:ShiftConvolve}
\end{algorithm}

To conserve runtime and accuracy in practice ShiftConvolve treats the special degenerate cases $p_i \in\{0,1\}$ differently.
That is, before any convolutions are done on the input vector, each $p_i=0$ is discarded (these correspond to guaranteed failures,
or adding the constant random variable $0$) and each $p_i=1$ is also discarded, but represents a (+1) shift in index for the final pmf
(guaranteed successes, or adding the constant random variable $1$).

Finally, the ShiftConvolvePoibin package offers the user the option of forgoing any exponential shift if one wants to compute the entire pmf
rather than a tail.

\section{Comparative analysis}

In this section we look at the performance of the exact methods DFT-CF,
DC-FFT and ShiftConvolvePoibin, as well as the approximation method SA.
We start with the analysis of the accuracy of the computed right tail
probability.

\subsection{Accuracy}

We first look only at SA and ShiftConvolve. Panel B of Figure \ref{fig:SP-problem}
confirms that SA offers a good relative accuracy for most of the range
of values $s_{0}$ for which the tail probability is small (panel
A), however, for extremely large values of $s_{0}$ the relative accuracy
is compromised. More alarming is the fact that for values of $s_{0}<4000$
SA consistently and inaccurately reports very small tail probabilities
instead of values close to 1. In addition, SA occasionally reports
probabilities that are larger than 1.

\begin{figure}
\centering %
\begin{tabular}{ll}
A. & B.\tabularnewline
\includegraphics[width=3in]{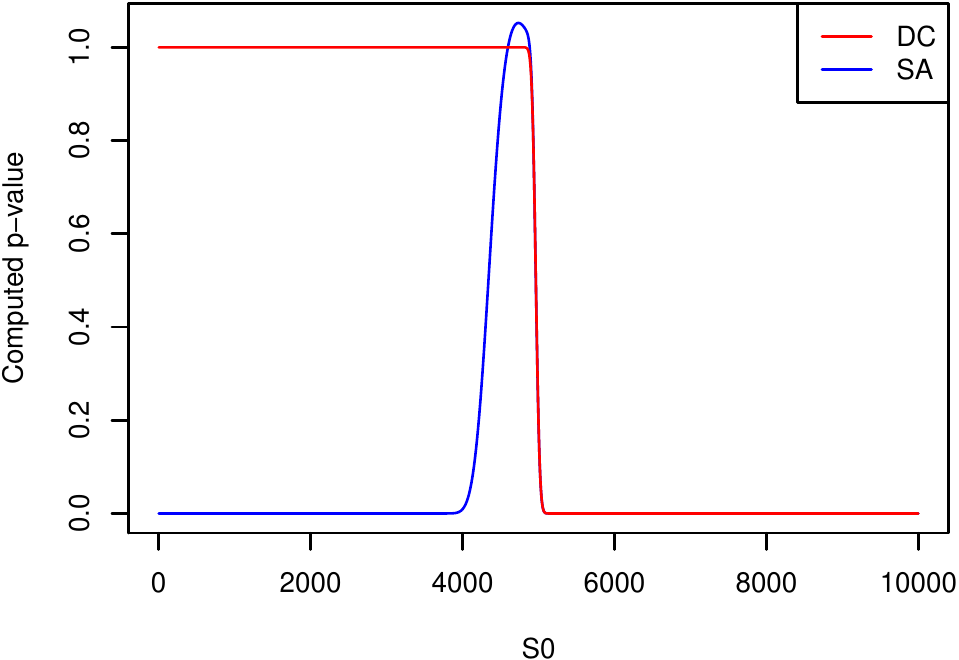} & \includegraphics[width=3in]{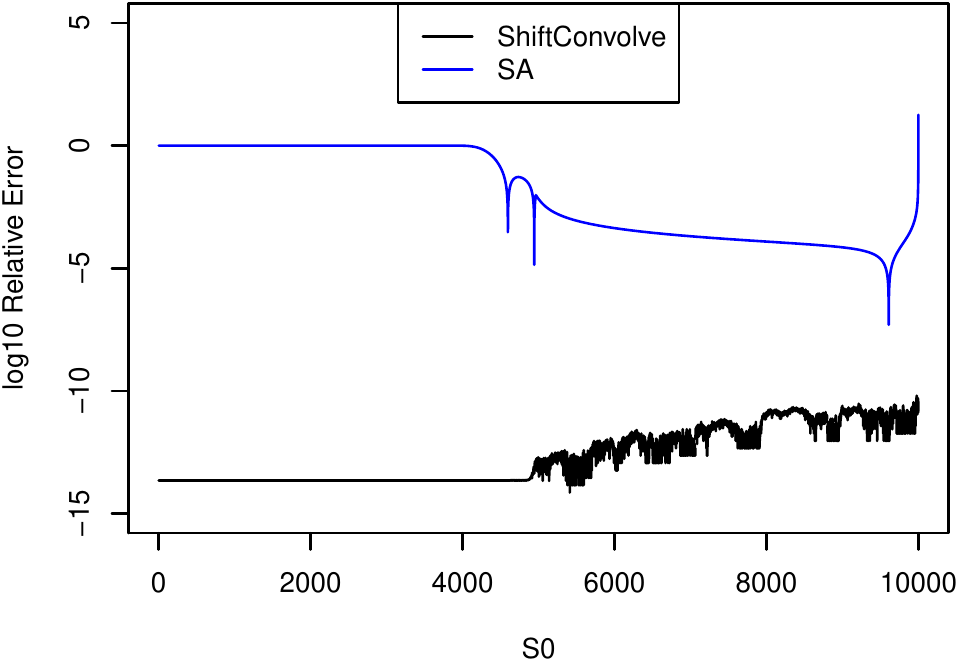}\tabularnewline
\end{tabular}\caption{\textbf{Failures of SA.} The panels highlight regions where SA fails
to accurately compute the right tail probability. Here, $N = 10000$ and the $p_i$ were sampled uniformly from $[0,1]$. (A) The computed right tail probability
(p-value) reported by SA and the accurate DC. (B) The corresponding
(log base 10 of) the relative error where we added for reference the
significantly smaller relative error of ShiftConvolve. The performance
of SA is of particular concern for $s_{0}\le4000$ as it reports p-values
that are smaller than 0.1 giving the user no indication that it might
be off from the correct tail probability which is close to 1. \label{fig:SP-problem}}
\end{figure}

Figures \ref{fig:right-tails} and \ref{fig:right-tail-error} as
well as Supplementary Figures \ref{fig:right-tail-2} and \ref{fig:right-tail-error-2}
confirm what we noted above: DFT-CF and DC-FFT cannot recover entries
that are smaller than approximately $10^{-16}$. SA's performance in
those figures is consistent with Figure \ref{fig:SP-problem} analyzed
above: its accuracy is compromised as we are close to the maximal
possible value and it is significantly off when the p-value is close
to 1. In contrast, ShiftConvolve retains very good accuracy (10 accurate
digits or more) throughout the entire range of values.

We also note that for the more skewed Beta(3, 0.1) distribution with higher proportions of 1s present (due to roundoff errors),
the SA algorithm failed to run, returning an error (as in panel C of Supplementary Figure \ref{fig:right-tail-2}, where the blue curve is absent).
Similarly, DC-FFT breaks when $N$ is very large ($N\ge5.5\cdot10^5$).

\begin{figure}
\centering %
\begin{tabular}{ll}
A. & B.\tabularnewline
\includegraphics[width=3in]{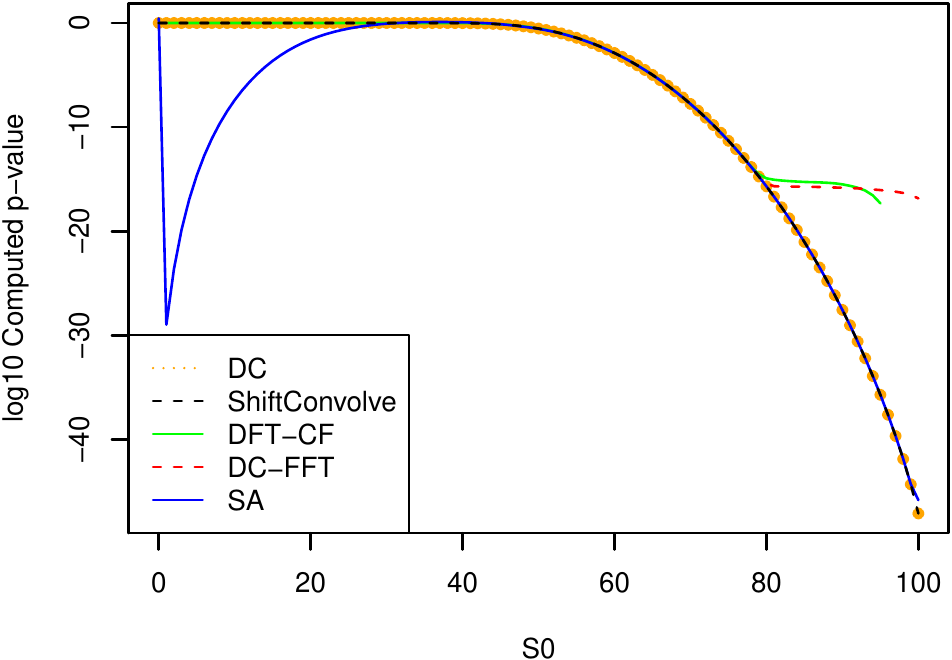} & \includegraphics[width=3in]{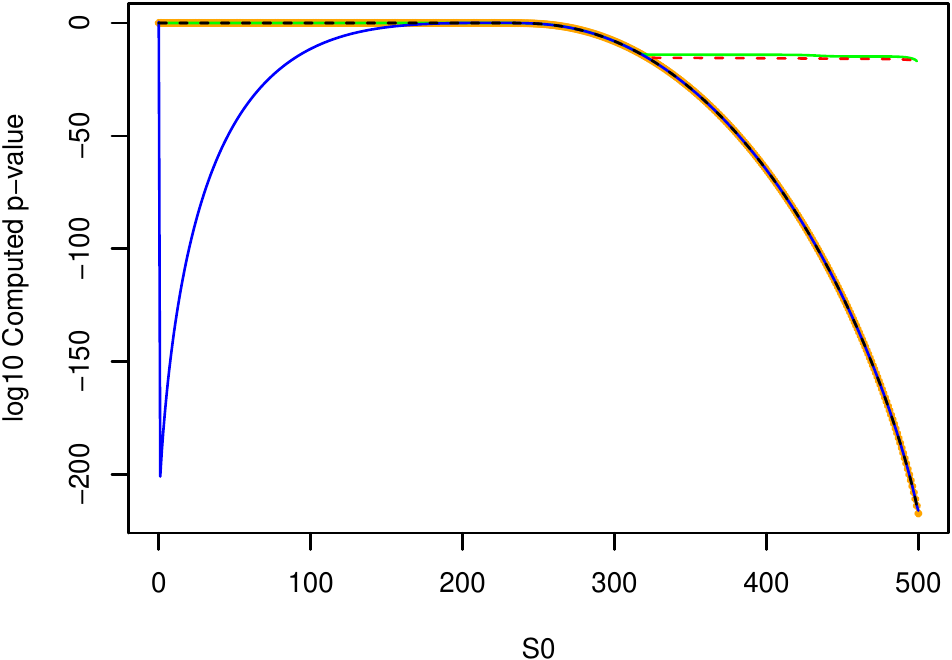}\tabularnewline
C. & D.\tabularnewline
\includegraphics[width=3in]{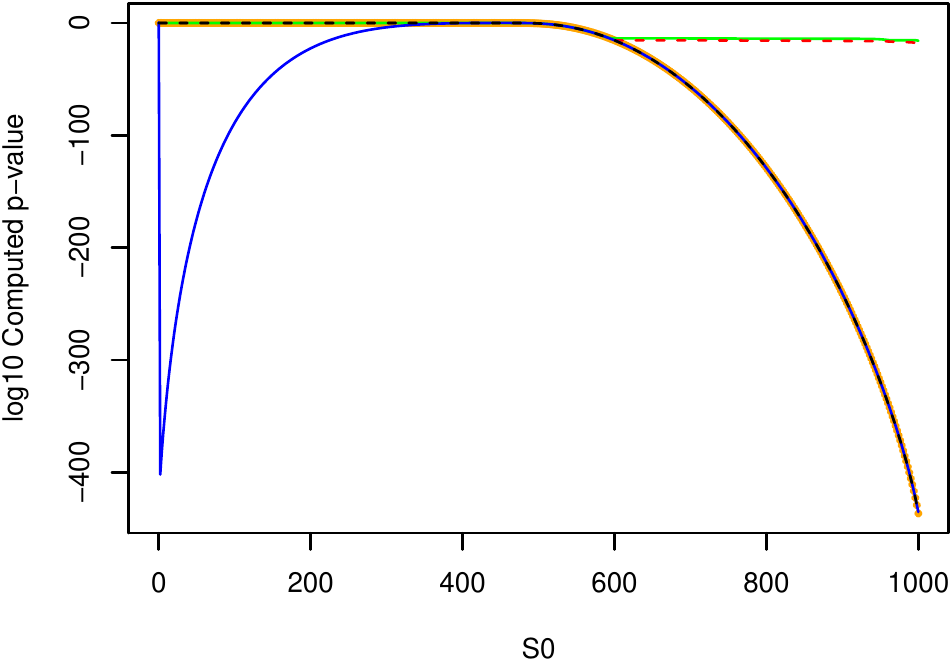} & \includegraphics[width=3in]{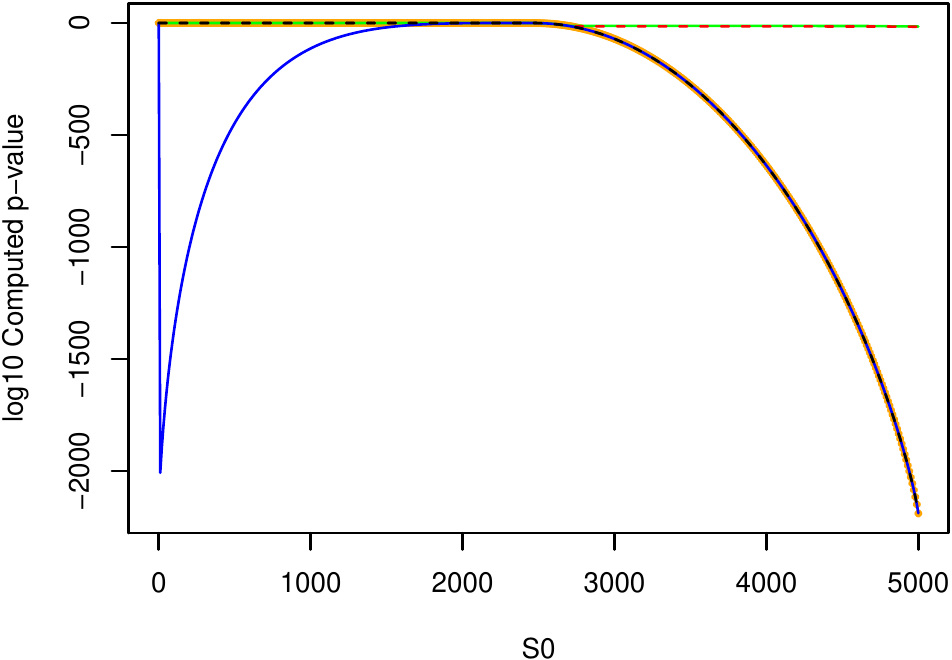}\tabularnewline
E. & F.\tabularnewline
\includegraphics[width=3in]{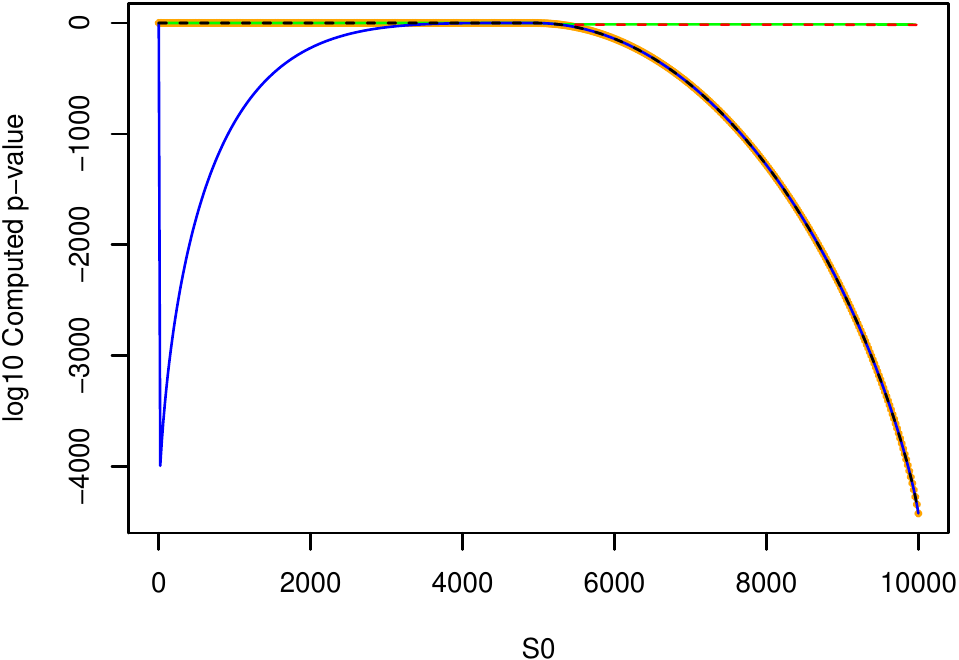} & \includegraphics[width=3in]{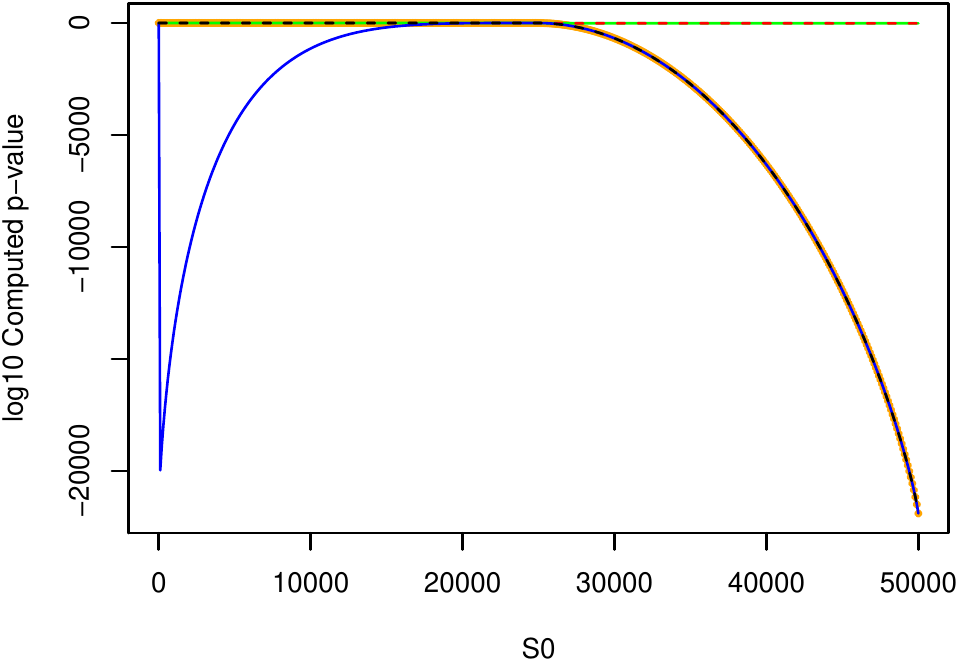}\tabularnewline
\end{tabular}\caption{\textbf{Right tail probability.} The panels show the (log of the)
computed right tail probability (p-value) for varying settings of
the parameters. For the corresponding relative errors see Figure \ref{fig:right-tail-error}.
(A) $N=100$, $p_i$ sampled uniformly. DFT-CF does not always
extend across the entire range of possible values of the statistic
because it reports 0s (log is $-\infty$ in some cases). (B) $N=500$, $p_i$ sampled uniformly. (C) $N=1000$, $p_i$ sampled uniformly.
(D) $N=5000$, $p_i$ sampled uniformly. (E) $N=10000$, $p_i$ sampled uniformly. (F) $N=50000$, $p_i$ sampled uniformly. \label{fig:right-tails}}
\end{figure}

\begin{figure}
\centering %
\begin{tabular}{ll}
A.  & B.\tabularnewline
\includegraphics[width=3in]{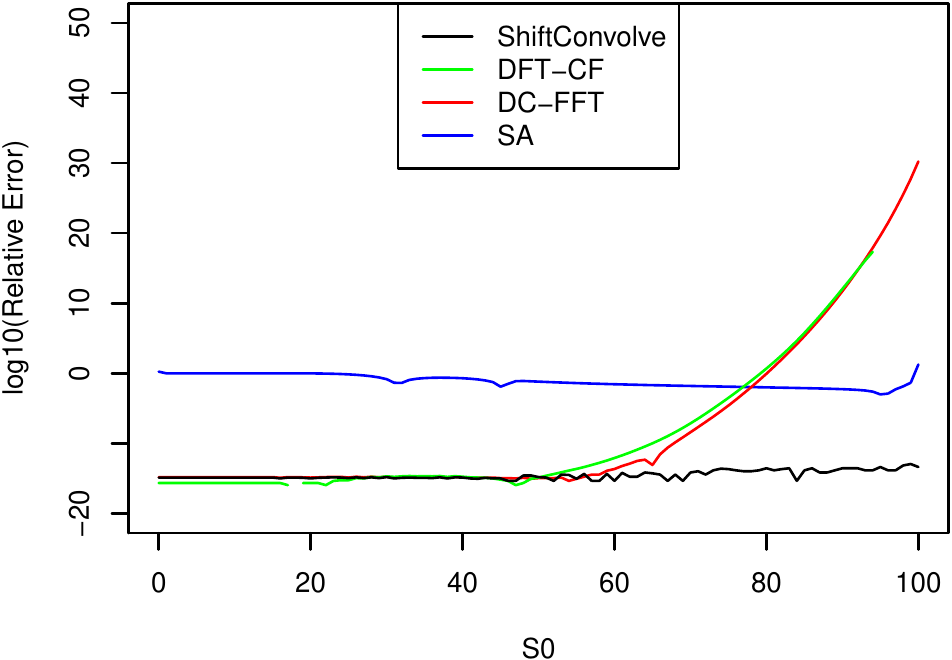} & \includegraphics[width=3in]{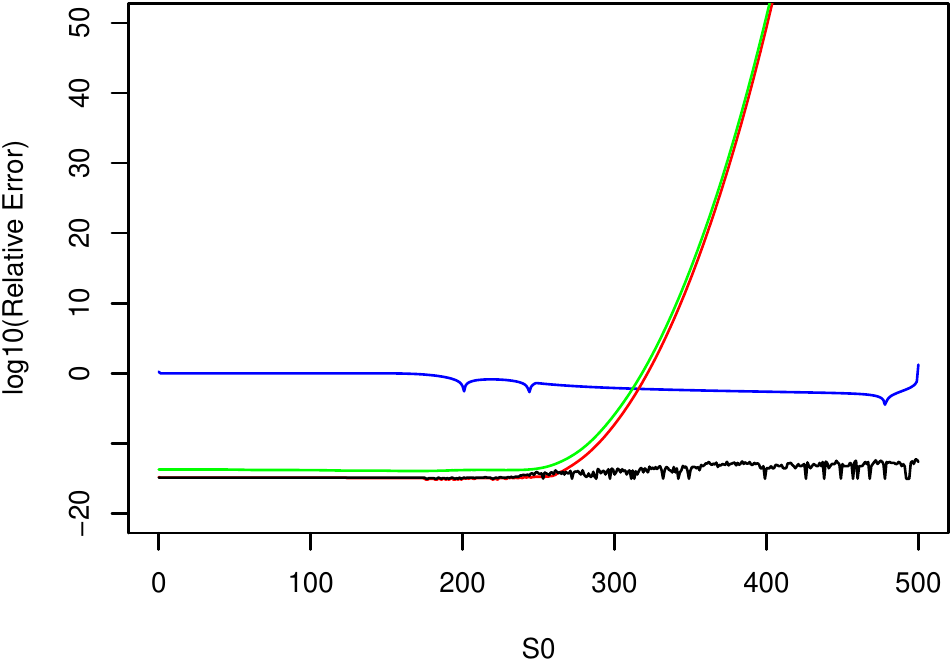}\tabularnewline
C. & D.\tabularnewline
\includegraphics[width=3in]{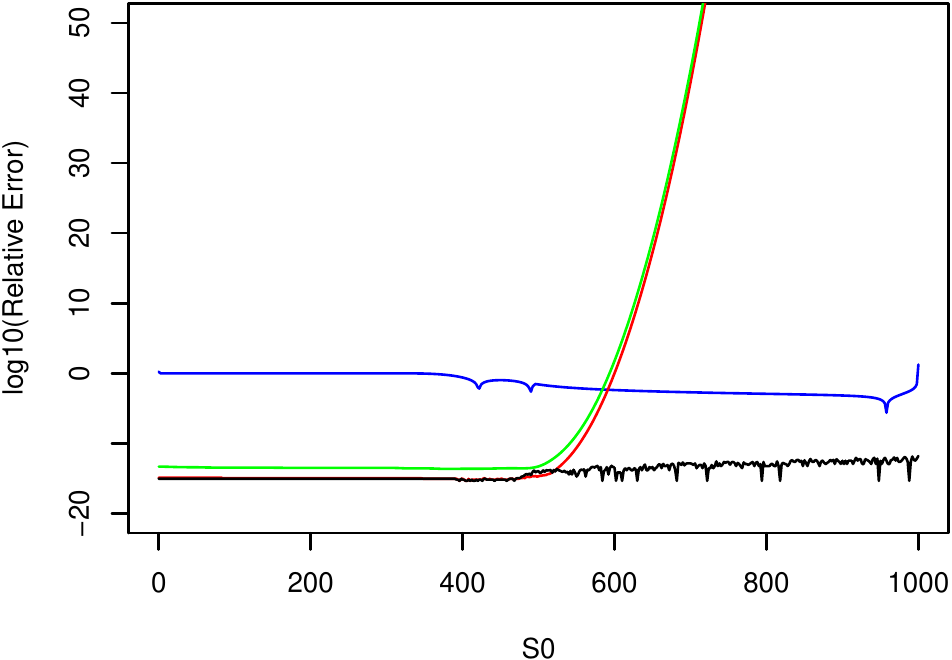} & \includegraphics[width=3in]{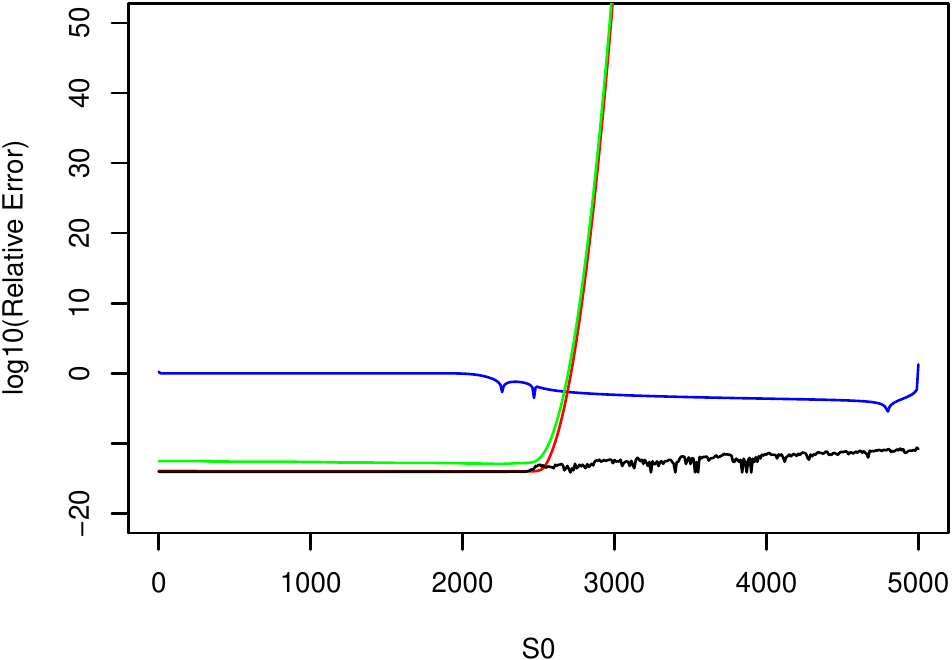}\tabularnewline
E. & F.\tabularnewline
\includegraphics[width=3in]{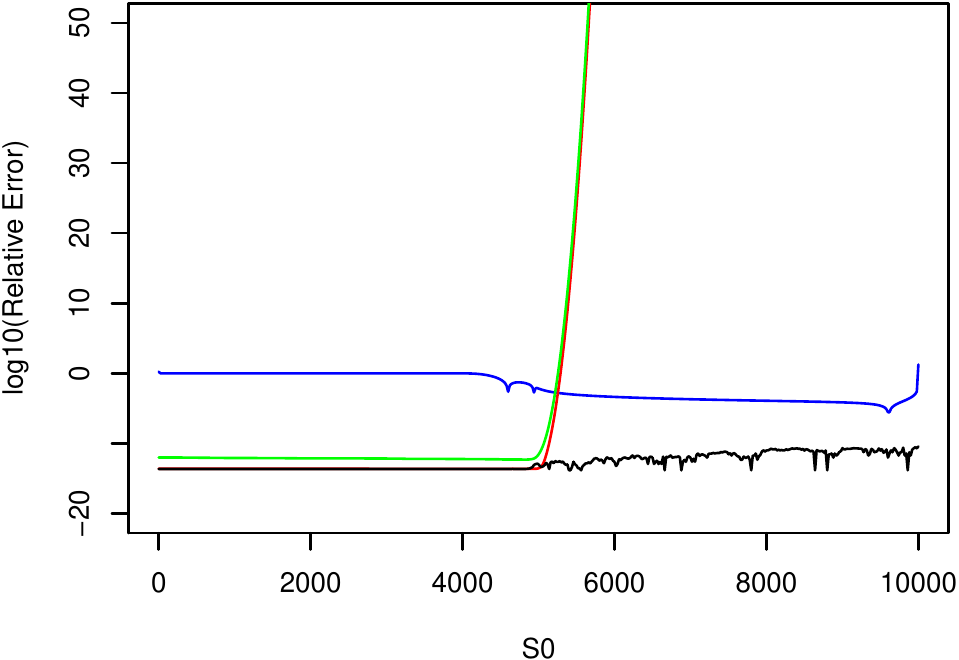} & \includegraphics[width=3in]{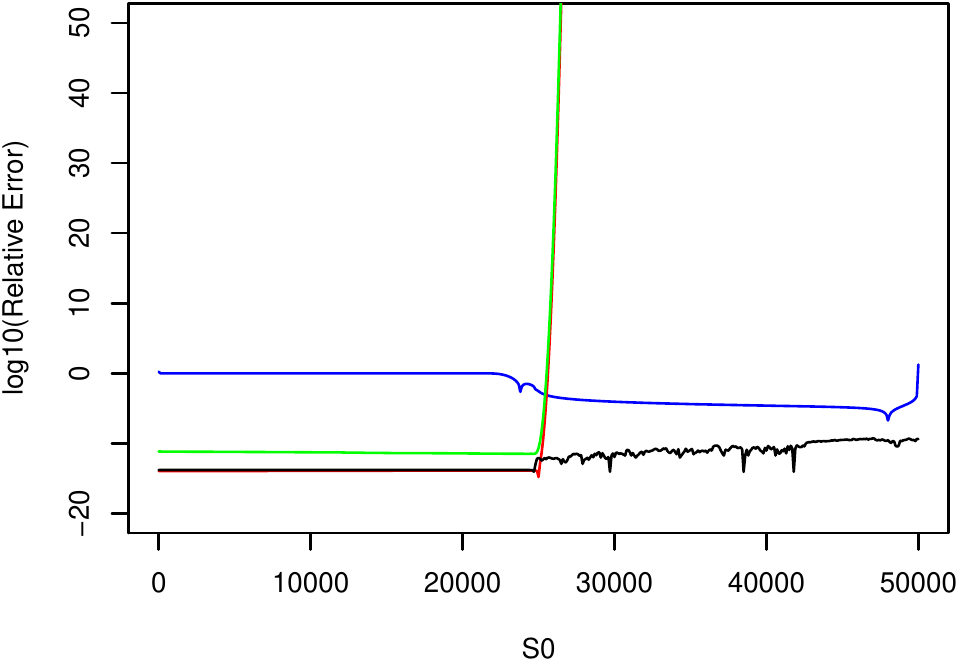}\tabularnewline
\end{tabular}\caption{\textbf{Gauging the accuracy.} The panels show the (log of the) relative
error (with DC taken as the gold standard) in computing the right tail probability for varying settings
of the parameters. See Figure \ref{fig:right-tails} for the right
tail probability itself. (A) $N=100$, $p_i$ sampled uniformly. (B) $N=500$, $p_i$ sampled uniformly. (C) $N=1000$, $p_i$ sampled uniformly.
(D) $N=5000$, $p_i$ sampled uniformly. (E) $N=10000$, $p_i$ sampled uniformly. (F) $N=50000$, $p_i$ sampled uniformly. \label{fig:right-tail-error}}
\end{figure}

\subsection{Complexity / runtime\label{subsec:Complexity}}

The memory and runtime complexities of an algorithm are important
practical considerations. In our case the memory complexity of all
the procedures is linear, or $O(N)$, but their runtimes differ considerably.

Starting with the approximation methods, they all require computing
at least a couple of moments of the distribution hence they are typically
$O(N)$, which is clearly the case for the normal approximation, as
well as the refined normal approximation (RNA) mentioned in \cite{Hong2013},
but it also applies to the more computationally-intensive SA. The
computationally-demanding part of SA is finding the desired exponential
shift $\tht$ but in practice this requires evaluating the MGF at
only a small number of candidate values of $\tht$ so the number of
operations of this step is still $O(N)$. This is borne out empirically
in panel B of Figure \ref{fig:timing-all} where the execution time of SA follows a line of slope 1 in log-log scale.

As for the exact methods, the runtime complexity of DC is $O\left(N^{2}\right)$:
there are $N$ steps (the outer loop of Algorithm 1 in \cite{Biscarri2018})
and the number of operations in the $k$th step is $O(k)$. Moving
on to DFT-CF, as noted by Madsen et al., its runtime complexity is
also $O\left(N^{2}\right)$ because this is the complexity of the
part of the procedure that computes the characteristic function (CF)
and inverting it is done using FFT in $O\left(N\log N\right)$.
Again, these quadratic complexities can be observed in panel B of Figure \ref{fig:timing-all}
where the execution times of DFT-CF and DC (we have both a version that works in log-space as well as the default version)
follow a line of slope 2 in log-log scale.

\begin{figure}
\centering %
\begin{tabular}{ll}
A. & B.\tabularnewline
\includegraphics[width=3in]{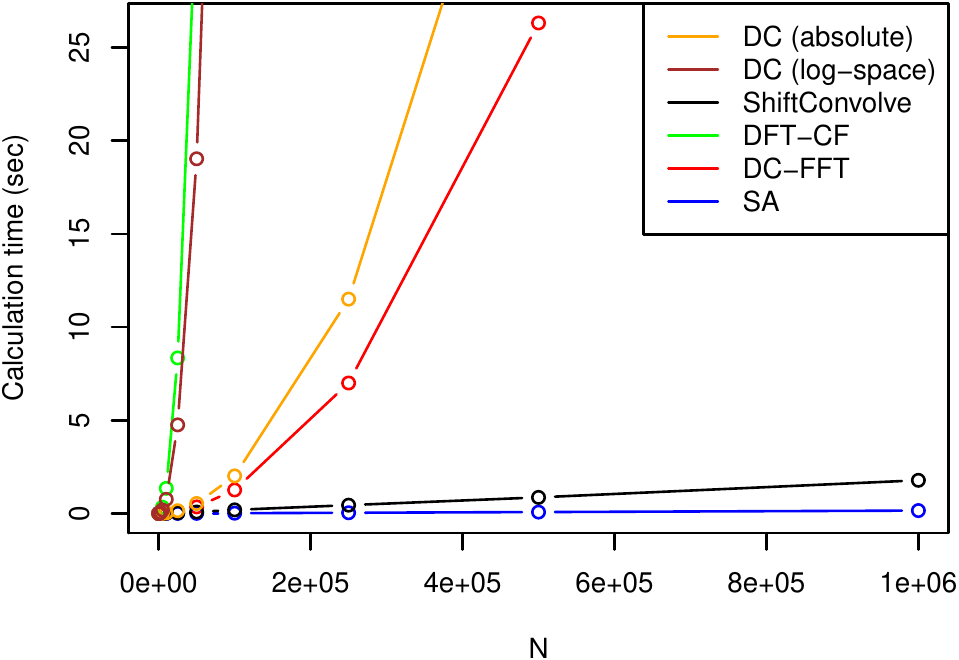} & \includegraphics[width=3in]{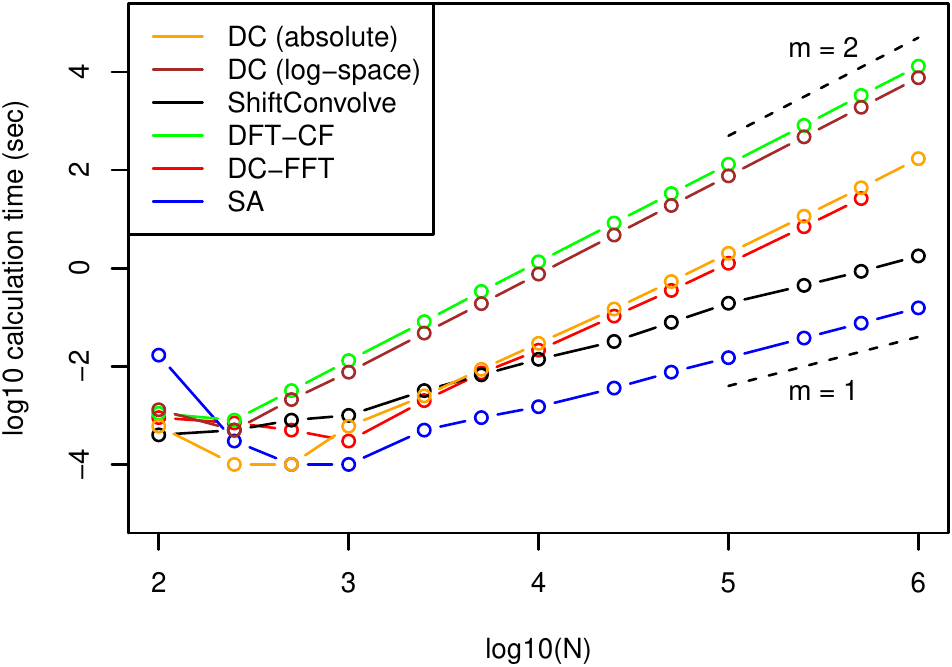}\tabularnewline
\end{tabular}\caption{\textbf{Runtime comparison.} The panels yield the runtime of the various
procedures we consider (panel A is in natural scale whereas panel B is in log-log scale).
DC-FFT was run in its default setting of $M=2$ (the missing point for $N=10^6$ is due to the its failure on that input size).
Each data point represents the average runtime over 10 applications
of the procedure each using an independently and uniformly drawn $N$-dimensional
vector of Bernoulli success probabilities (different vector for each of the 10 applications).
DC (absolute) is Algorithm 1 of Biscarri et al.~and DC (log-space) is a version of DC that is part of the ShiftConvolvePoibin package
and that works with logs to extend DC's dynamic range.
In panel B we added short line segments of slopes $m=1$ and $m=2$ for reference. \label{fig:timing-all}}
\end{figure}

Biscarri et al.~empirically demonstrated that DC-FFT is significantly
faster than DFT-CF but did not offer a complexity analysis of DC-FFT
and in fact its runtime complexity varies with the parameter $M$:
with $M=N$ the runtime is the same $O\left(N\left(\log N\right)^{2}\right)$
as ShiftConvolve's (see our analysis below) but with its default setting
of $M=2$, or any other constant, it is actually $O\left(N^{2}\right)$ (as can be observed in Figure \ref{fig:timing-all}).
Biscarri et al.~also suggest a heuristic of choosing $M=\max\left\{ 2,2^{\left[\log_{2}\left(N/750\right)\right]}\right\} $,
where $\left[x\right]$ is the round-to-nearest-integer function which
again has the same asymptotic complexity of $O\left(N\left(\log N\right)^{2}\right)$
as when $M=N$, although in practice it can be much faster for any
given $N$ (Supplementary Figure \ref{fig:runtime-Biscarri}).

The runtime complexity of our ShiftConvolve is $O\left(N\left(\log N\right)^{2}\right)$.
Indeed, assuming that $N=2^{K}$ there are $K$ steps each of which consists of
using FFT to convolve pairs of equal-length vectors: at the $k$th
step there are $N/2^{k}$ pairwise convolutions of $2^{k}$-dimensional
vectors. Hence the number of operations in the $k$th step is $O\left(N/2^{k}\cdot2^{k}\log2^{k}\right)=O\left(N\cdot k\right)$
and summing over $k=1,\dots,K$ we get the stated $O\left(N\cdot K^{2}\right)$.

Note that the modifications that we made to speed up Algorithm~\ref{alg:DC-FFT-M-N}
do not change its runtime complexity. Rather, the improvement is in
the constant. Specifically, in Algorithm~\ref{alg:DC-FFT-M-N} we
compute the $n$-dimensional inverse DFT $\uu=D^{-1}\left(\vv_{i}\odot\vv_{j}\right)$ (line 12),
as well as the $2n$-dimensional DFT of the padded $\uu$ (line 14).
Compare that with applying the same $n$-dimensional inverse DFT $\uu=D^{-1}\left(\vv_{i}\odot\vv_{j}\right)$ in Algorithm~\ref{alg:more-efficient}
(line 13), as well as the $n$-dimensional (rather than $2n$-dimensional) DFT of $\uu\odot\omgg$ (line 14).
The former requires about 50\% more work, hence ShiftConvolve's improvement offers a reduction
of about 50\% in the runtime for large $N$.
% compute the $2n$-dimensional DFTs $\vv^{*}_{i}=D\uu_{i}$ and $\vv^{*}_{j}=D\uu_{j}$ (lines 3 and 14),
% where $n$ is the unpadded length of $\uu_{i}$ and $\uu_{j}$, as well as the
% $2n$ dimensional inverse DFT $\uu=D^{-1}\left(\vv_{i}\odot\vv_{j}\right)$ (lines 12 and 18).
% Compare that with applying two $n$-dimensional DFTs to recover just
% the odd entries of the $\ww_{i}$ and $\ww_{j}$ in addition again
% to the $2n$-dimensional $D^{-1}\left(\ww_{i}\odot\ww_{j}\right)$
% we compute in Algorithm~\ref{alg:more-efficient}. The former requires
% about 50\% more work, hence ShitfConvolve's improvement offers a reduction
% of about 50\% in the runtime for large $N$.

ShiftConvolve differs from DC-FFT with its default setting of $M=N$ by
employing the exponential shift and its different approach
to FFT-based convolution outlined in Algorithm \ref{alg:more-efficient}.
While the two algorithms share the same runtime complexity we found that in practice
ShiftConvolve is typically significantly faster (Supplementary Figure \ref{fig:runtime-Biscarri}).
In fact, we observed that ShiftConvolve speed is comparable to 
DC-FFT's when using the above formula of $M=\max\left\{ 2,2^{\left[\log_{2}\left(N/750\right)\right]}\right\}$
(Supplementary Figure \ref{fig:runtime-Biscarri}).

Generally, the version of ShiftConvolve that uses minFFT is slightly faster than the version the uses FFTW,
hence we used the former version in our runtime benchmarks. All timing runs were executed on a 3.2GHz
Intel Core i7 MacMini with 32GB of RAM.

\section{Discussion}

Madsen et al.~pointed out that DFT-CF fails to accurately recover
small tail probabilities of the PBD resulting in extremely large relative
errors. As we show here, the same applies to DC-FFT, which is also an FFT-based
exact algorithm for computing the PBD.

As an alternative Madsen et al.~offer a saddle point approximation method (SA), which
does a significantly better job at recovering these small tail probabilities.
However, SA has its own accuracy issues as we approach the maximal
possible value of $N$, as well as when the right tail probability
is rather large. The latter is particularly troubling because the
typical user will not be aware that the small p-value that SA reports
is in fact very close to 1 (Figure \ref{fig:SP-problem}).

Following \cite{Keich2005,Keich2006,Nagarajan2009} our proposed
solution to this problem combines the same exponential shift (sometime
referred to as an exponential tilt) that the saddlepoint approximation
is based on with the FFT-based exact method.
Specifically, ShiftConvolvePoibin uses an exponential
shift combined with a souped-up version of the FFT-based convolution by aggregated
pairs approach as implemented in DC-FFT ($M=N$). The result is a
relatively fast exact algorithm that accurately computes tail probabilities across the entire
range of the PBD.

In terms of runtime complexity at $O\left(N\left(\log N\right)^{2}\right)$
ShiftConvolve is equivalent to DC-FFT ($M=N$), however both DFT-CF
as well as DC-FFT with its default setting of $M=2$ have a runtime
complexity of $O\left(N^{2}\right)$ which makes a significant difference
for large $N$.

It is worth noting that Biscarri et al.~recommend using the refined
normal approximation (RNA) for $N\ge10^{5}$, however Supplementary
Figure \ref{fig:RNA-problem} shows that as far as computing small
tail probabilities SA does a significantly better job than RNA and
both are inferior when compared with the accuracy of ShiftConvolve.

In terms of future research, while in practice ShiftConvolve is accurate
throughout the entire range of possible values it would be useful
to obtain upper bounds on its cumulative numerical error which can
then be compared with the computed result to guarantee its accuracy
(cf.~\cite{Wilson2017}). A related point is that ShiftConvolve
is designed to accurately recover the relevant section of the pmf
(panel B of Figure \ref{fig:DFT_error}) however it is worth noting
that a few well-selected shifts $\tht_{0}$ should allow us to accurately
recover the entire pmf. Regardless, the ShiftConvolvePoibin package
allows the user to compute the entire PBD while bypassing the
exponential shift. As such it is competitive with the fastest exact methods in this category.

The ShiftConvolvePoibin package is available to download from\\
\url{https://github.com/andrew12678/ShiftConvolve}
% https://github.com/andrew12678/ShiftConvolveFFTW

\section{Acknowledgements}

UK would like to thank Jakob Pedersen for initial discussions on this
problem as well as to Denzel Florez, Tingyue Liu, Xuanchi Liu, and
Buqing Yang for initial work on this project.

\newpage

\section{Supplementary Figures}

\begin{figure}
\centering %
\begin{tabular}{ll}
A. & B.\tabularnewline
\includegraphics[width=3in]{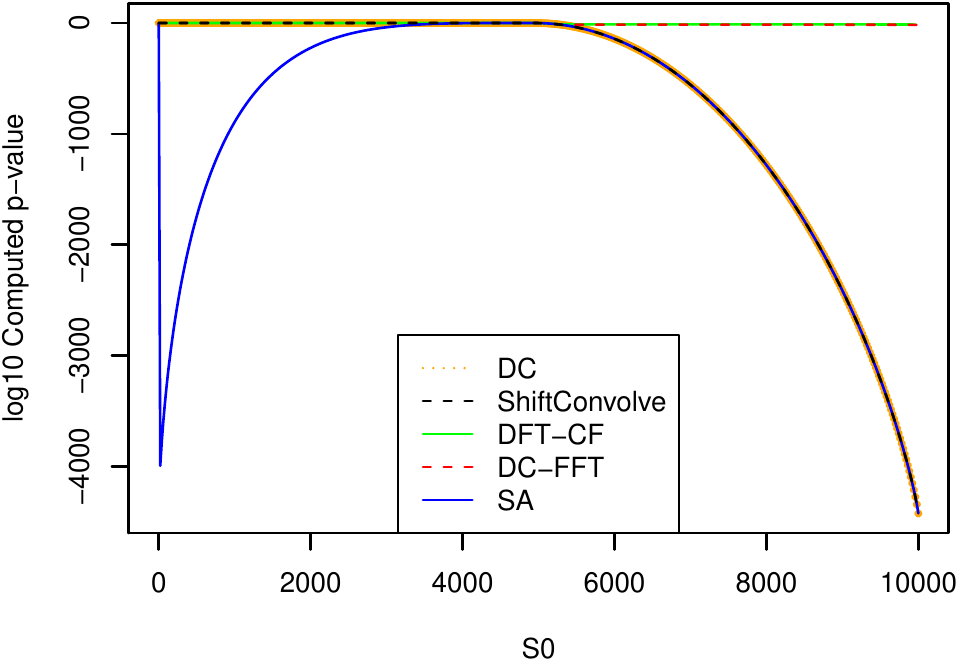} & \includegraphics[width=3in]{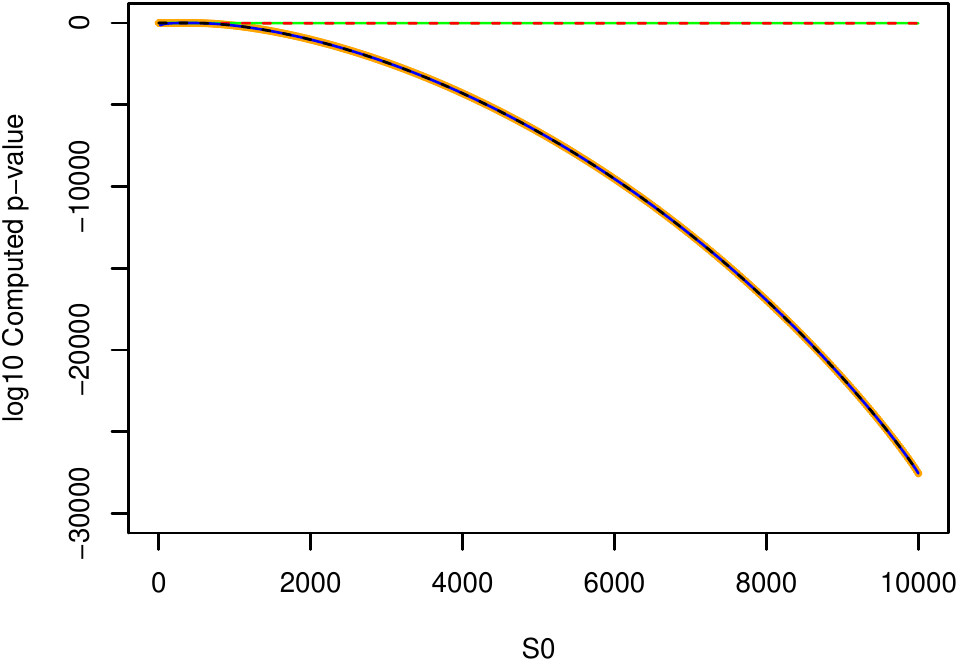}\tabularnewline
C. & D.\tabularnewline
\includegraphics[width=3in]{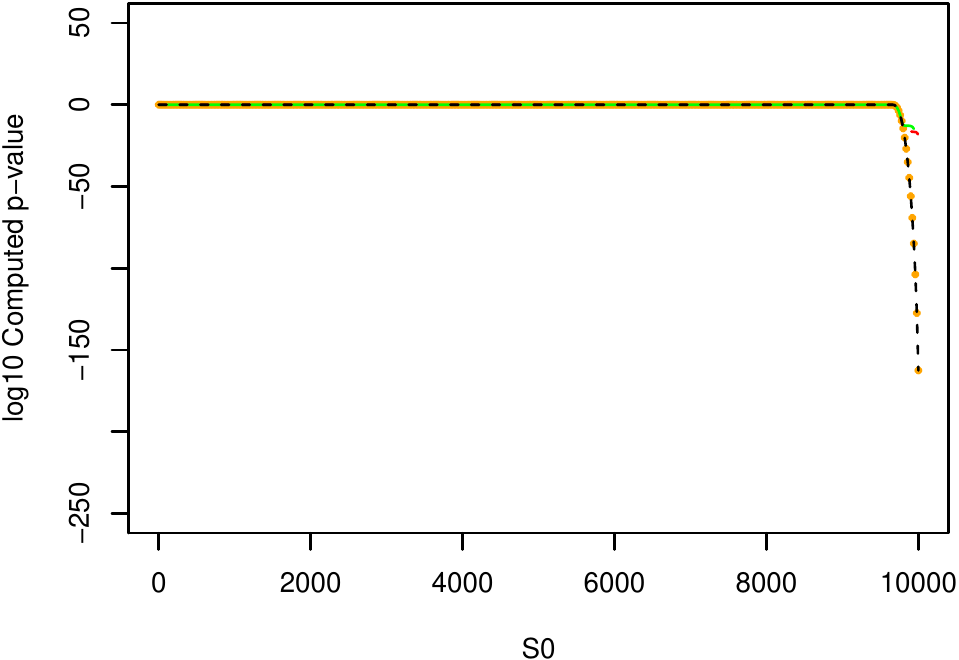} & \includegraphics[width=3in]{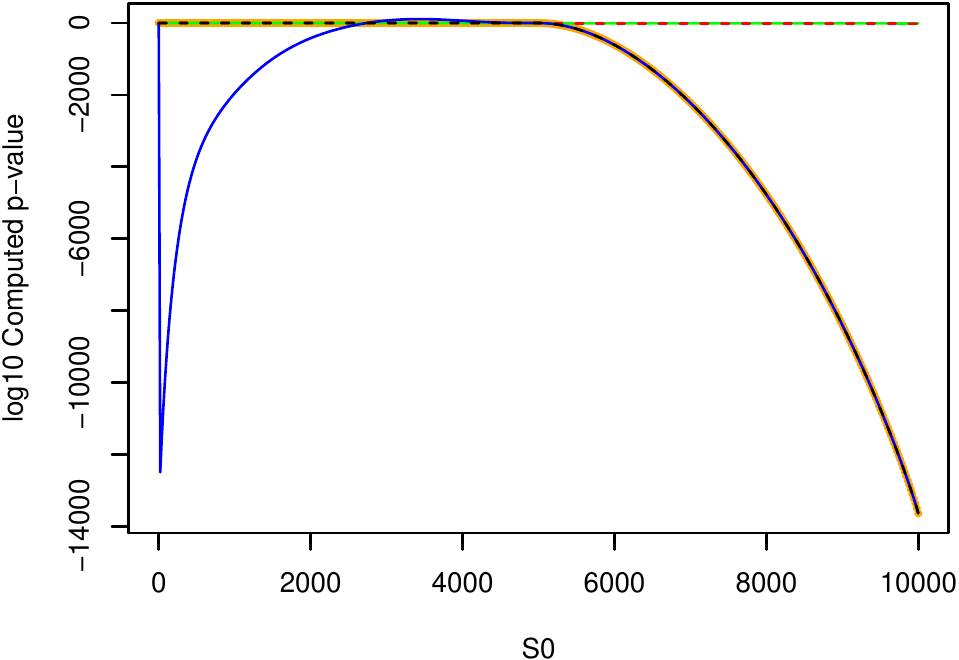}\tabularnewline
E. SA gets as high as 400 & F.\tabularnewline
\includegraphics[width=3in]{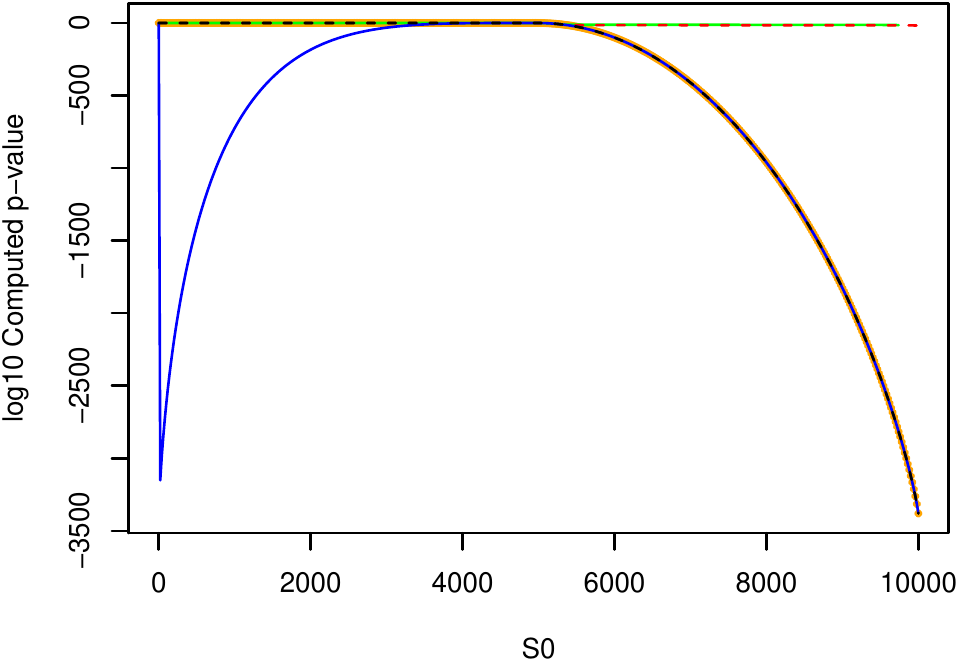} & \includegraphics[width=3in]{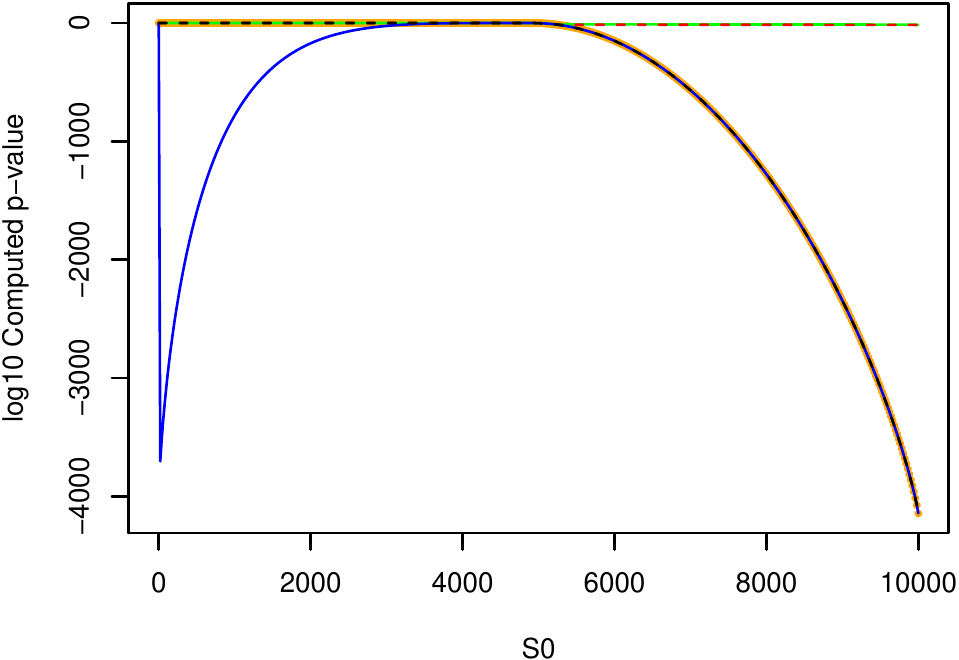}\tabularnewline
\end{tabular}\caption{\textbf{Right tail probability (II).} Similarly to Figure \ref{fig:right-tails}
the panels show the (log of the) computed right tail probability (p-value)
for additional settings of the parameters. For the corresponding relative
errors see Supplementary Figure \ref{fig:right-tail-error-2}. (A) $N = 10^4$, $p_i \sim U(0,1)$.
(B) $N = 10^4$, $p_i \sim \bd(0.1,3)$. (C) $N = 10^4$, $p_i \sim \bd(3,0.1)$.
Note that in this example with higher proportions of 1s present (due to roundoff errors)
the SA algorithm failed to run, returning an error where the blue curve is absent.
(D) $N = 10^4$, $p_i \sim \left[0.5\cdot\bd(3,0.1) + 0.5 \cdot\bd(0.1,3)\right]$ (E) $N = 10^4$, $p_i \sim \bd(3,3)$.
(F) $N = 10^4$, $p_i \sim \left[0.5\cdot\bd(3,10) + 0.5\cdot\bd(10,3)\right]$. \label{fig:right-tail-2}}
\end{figure}

\begin{figure}
\centering %
\begin{tabular}{ll}
A. & B.\tabularnewline
\includegraphics[width=3in]{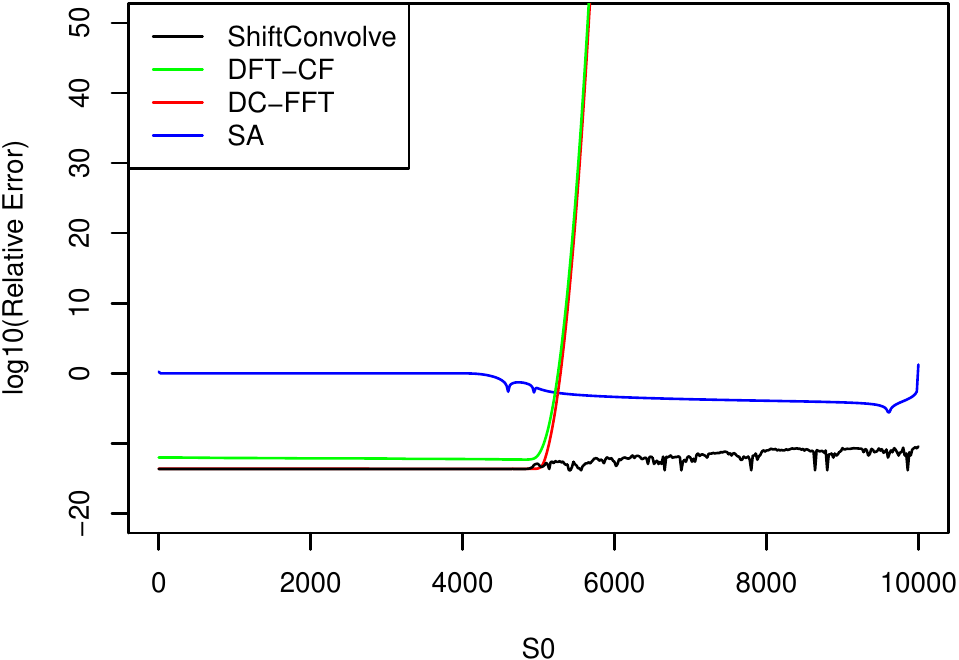} & \includegraphics[width=3in]{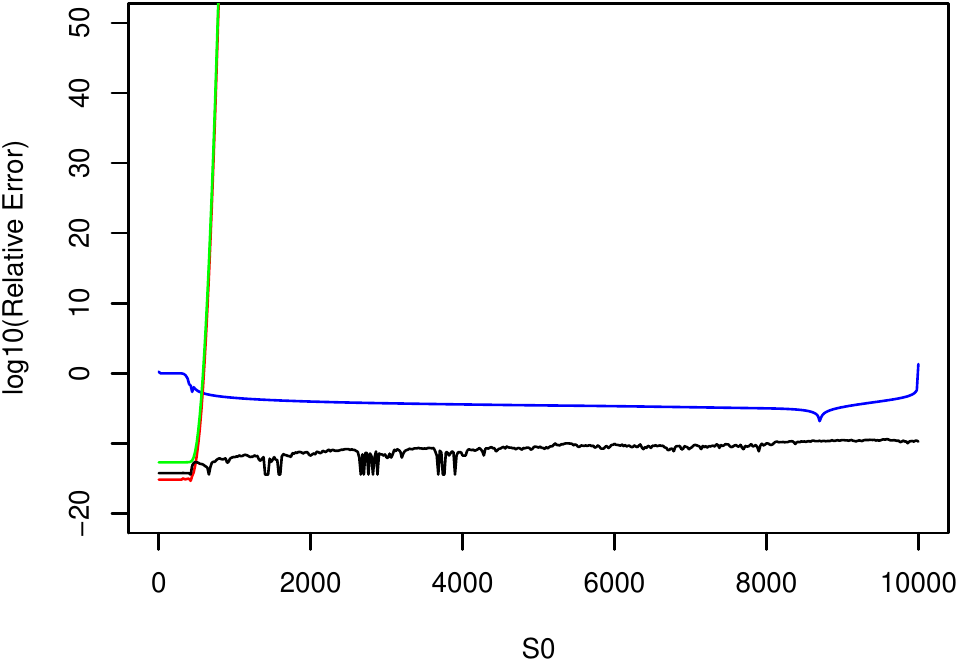}\tabularnewline
C. & D.\tabularnewline
\includegraphics[width=3in]{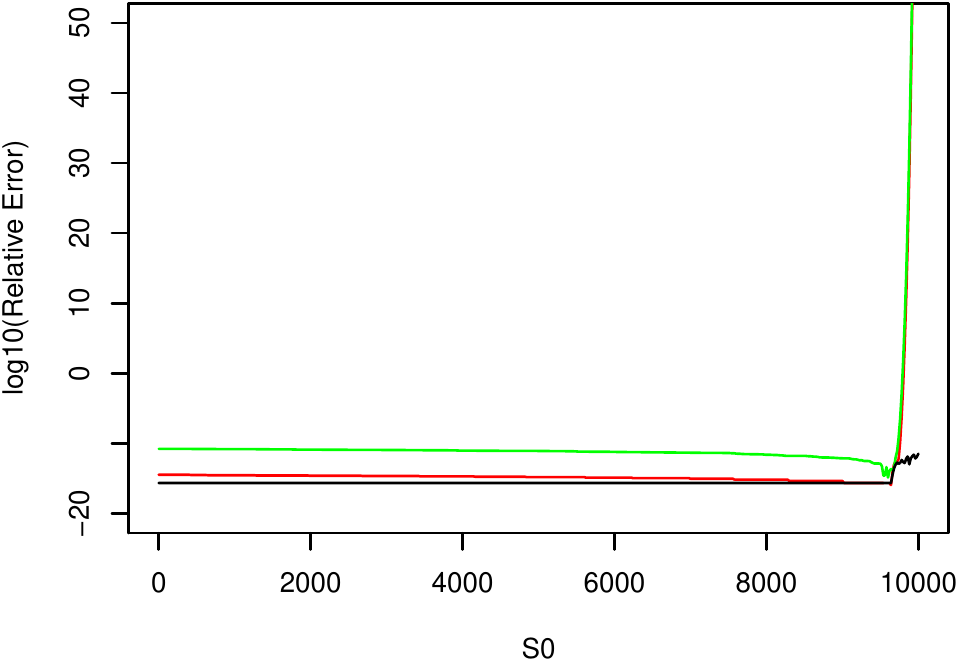} & \includegraphics[width=3in]{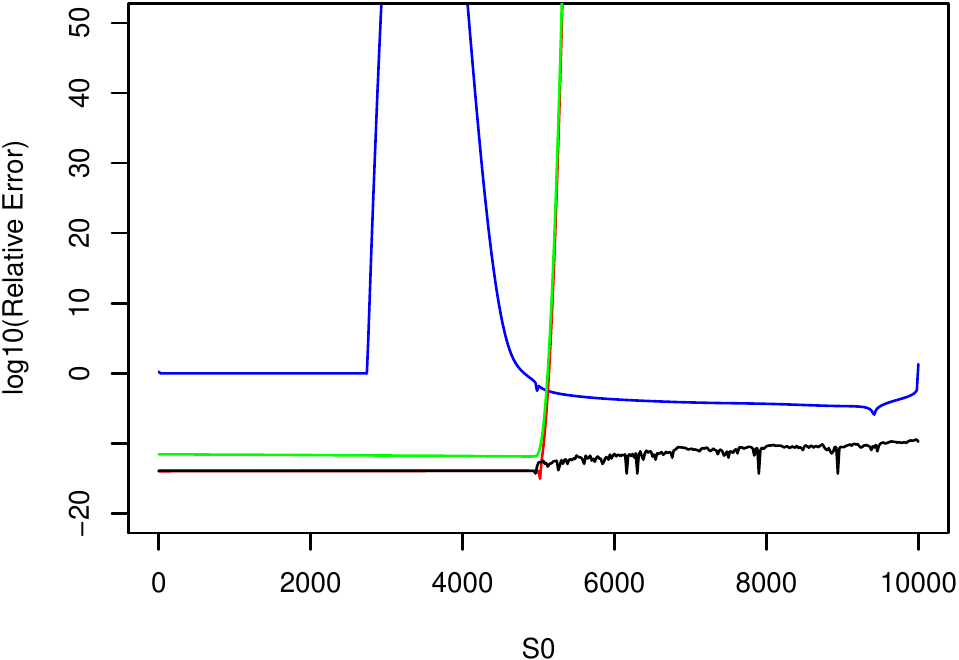}\tabularnewline
E. & F.\tabularnewline
\includegraphics[width=3in]{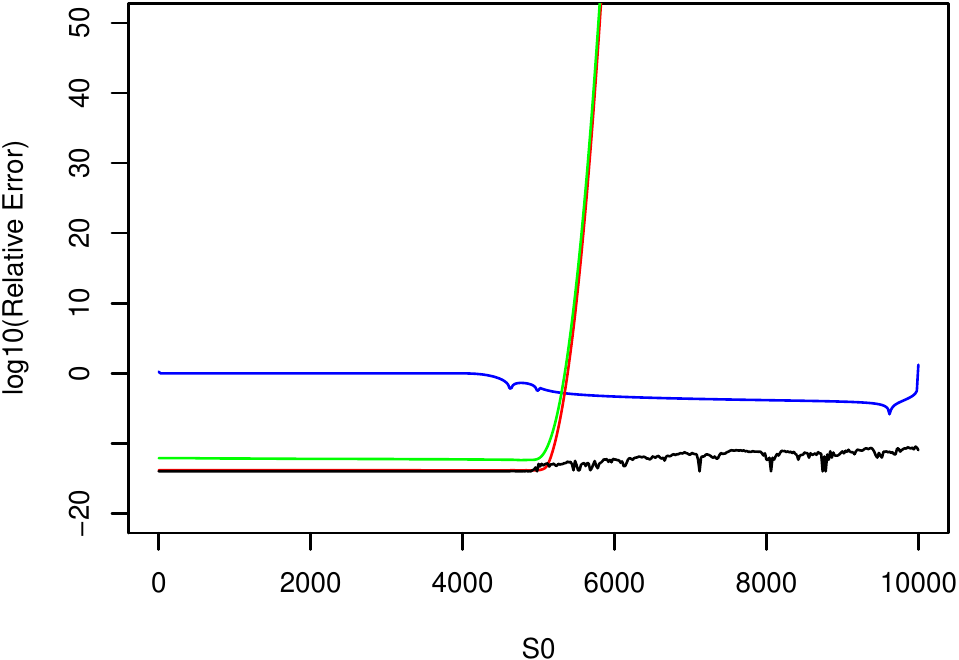} & \includegraphics[width=3in]{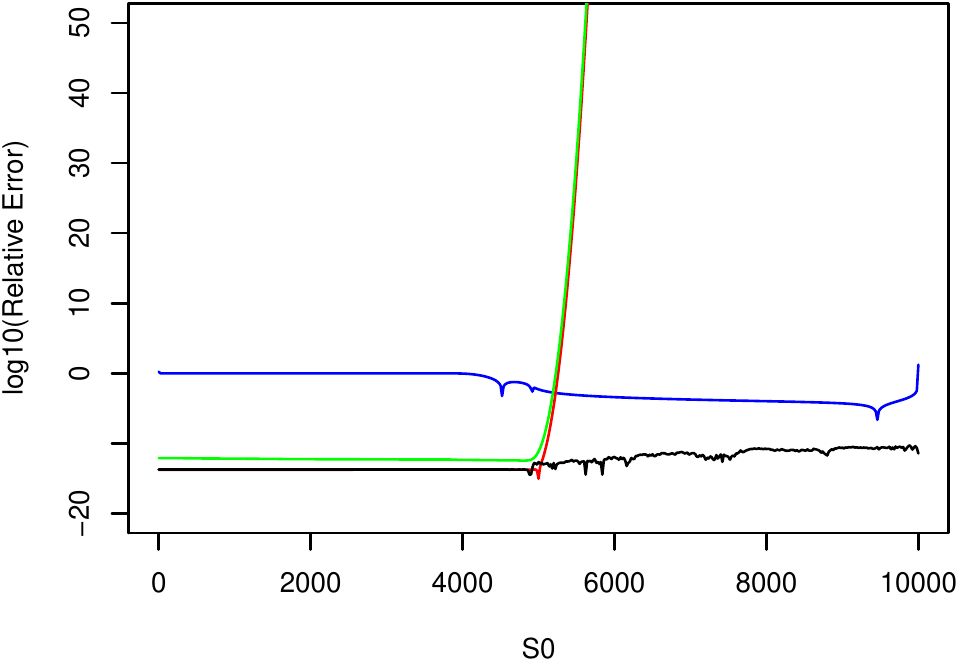}\tabularnewline
\end{tabular}\caption{\textbf{Gauging the accuracy.} Similarly to Figure \ref{fig:right-tail-error}
the panels show the (log of the) relative error in computing the right
tail\textbf{ }for varying settings of the parameters. See Figure \ref{fig:right-tail-2}
for the right tail probability itself. (A) $N = 10^4$, $p_i \sim U(0,1)$.
(B) $N = 10^4$, $p_i \sim \bd(0.1,3)$. (C) $N = 10^4$, $p_i \sim \bd(3,0.1)$.
(D) $N = 10^4$, $p_i \sim \left[0.5\cdot\bd(3,0.1) + 0.5 \cdot\bd(0.1,3)\right]$ (E) $N = 10^4$, $p_i \sim \bd(3,3)$.
(F) $N = 10^4$, $p_i \sim \left[0.5\cdot\bd(3,10) + 0.5\cdot\bd(10,3)\right]$.
\label{fig:right-tail-error-2}}
\end{figure}

\begin{figure}

\begin{centering}
\includegraphics[width=6in]{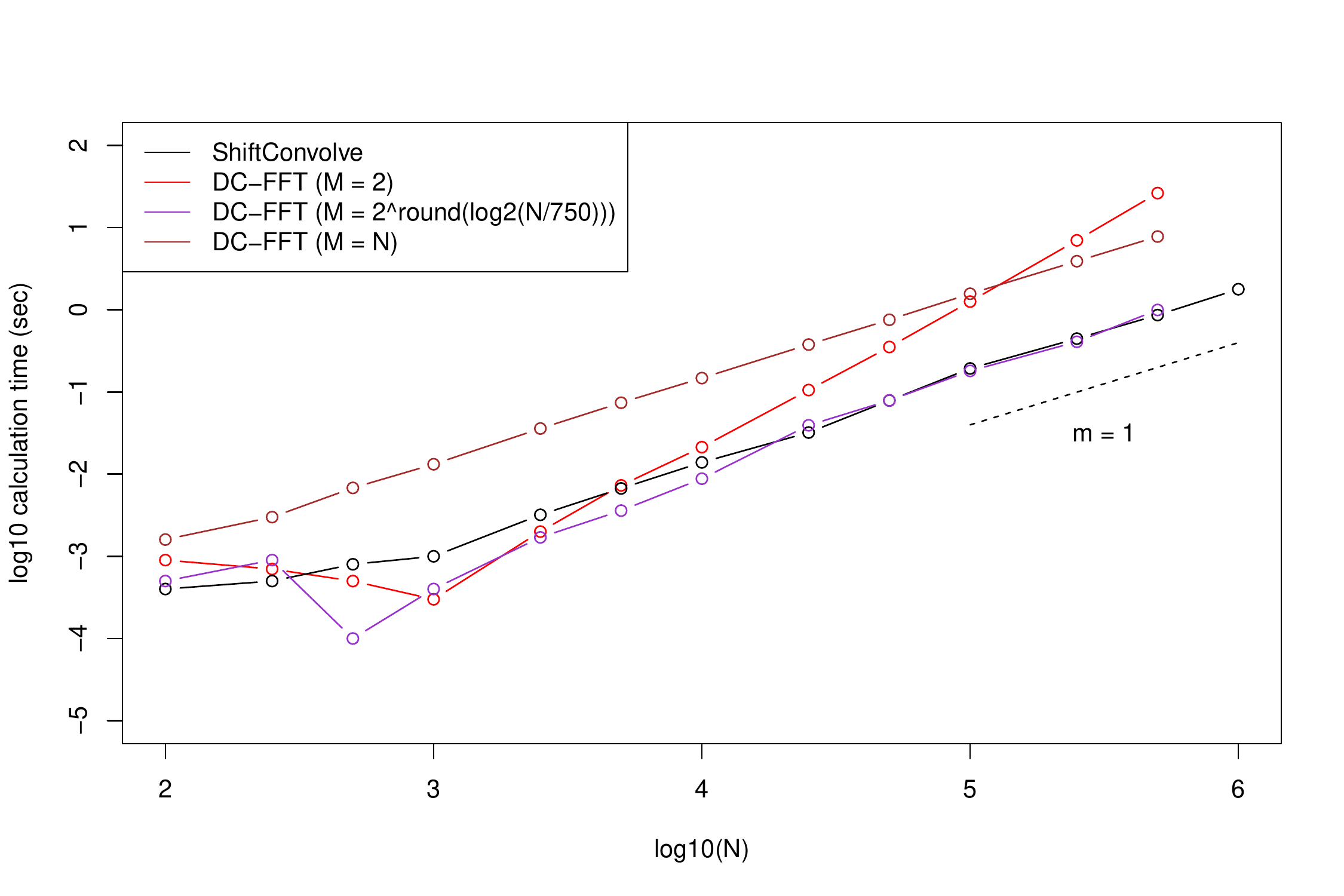}\caption{\textbf{Runtime: DC-FFT (varying $M$) and ShiftConvolve.} The runtime
of DC-FFT using different settings of the parameter $M$ and ShiftConvolve.
Unfortunately we could not get DC-FFT to work for $N\ge525\cdot10^{3}$.
A line segment of slope 1 was added for reference.\label{fig:runtime-Biscarri}}
\par\end{centering}
\end{figure}
\begin{figure}
\centering %
\begin{tabular}{ll}
A. & B.\tabularnewline
\includegraphics[width=3in]{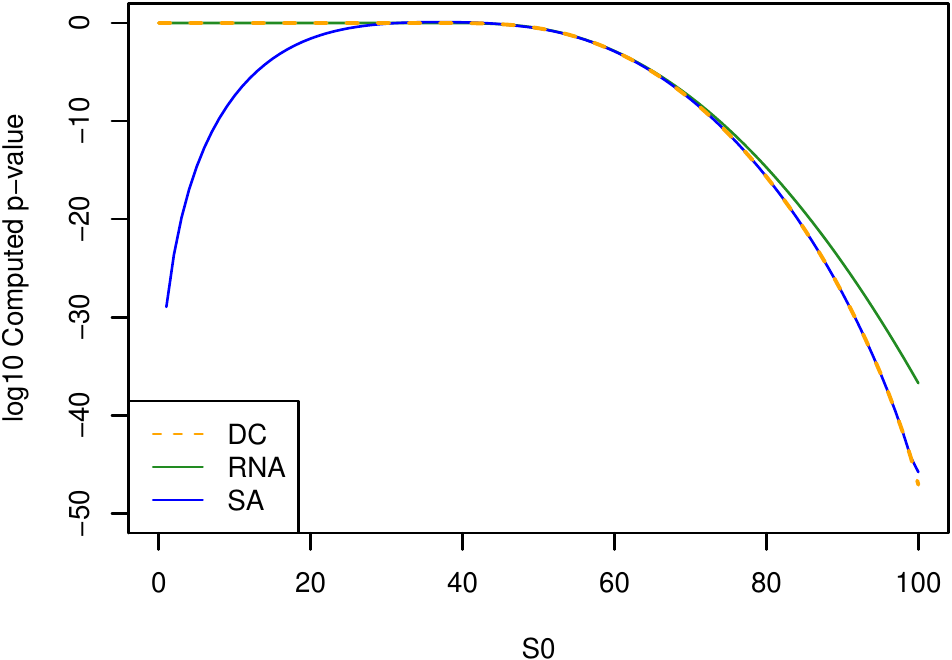} & \includegraphics[width=3in]{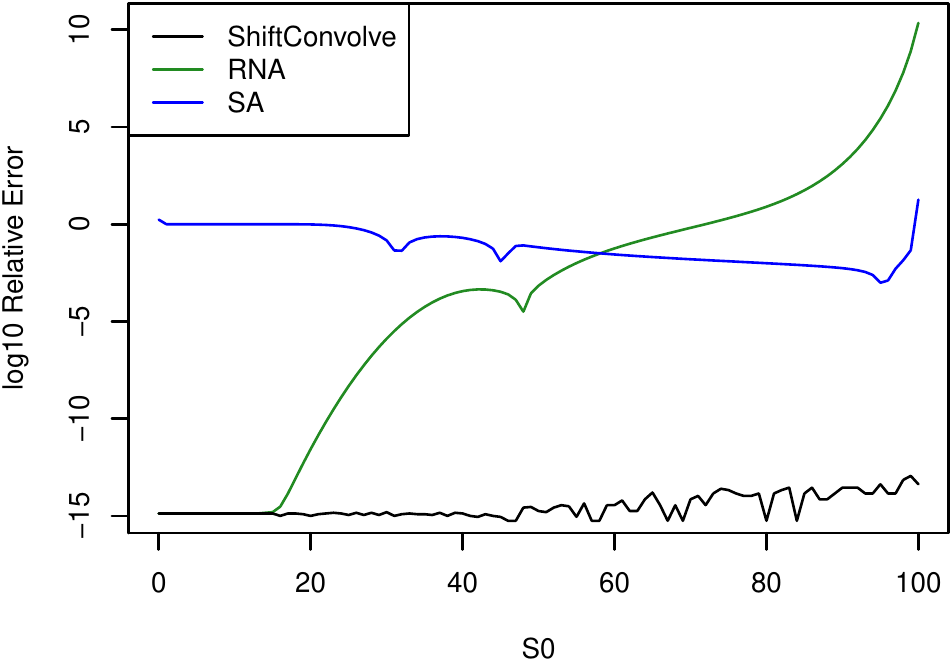}\tabularnewline
C. & D.\tabularnewline
\includegraphics[width=3in]{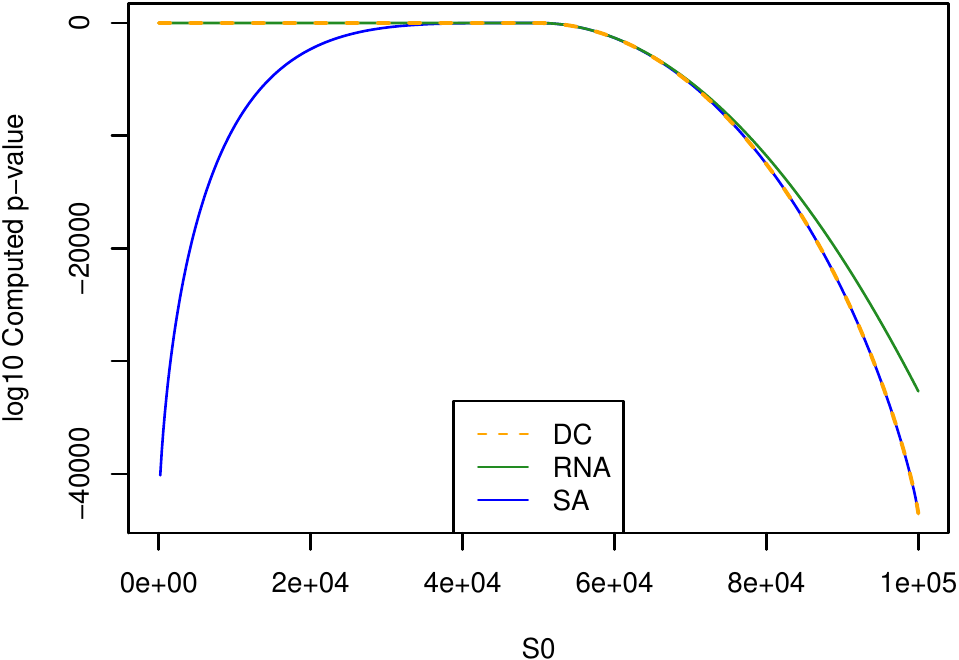} & \includegraphics[width=3in]{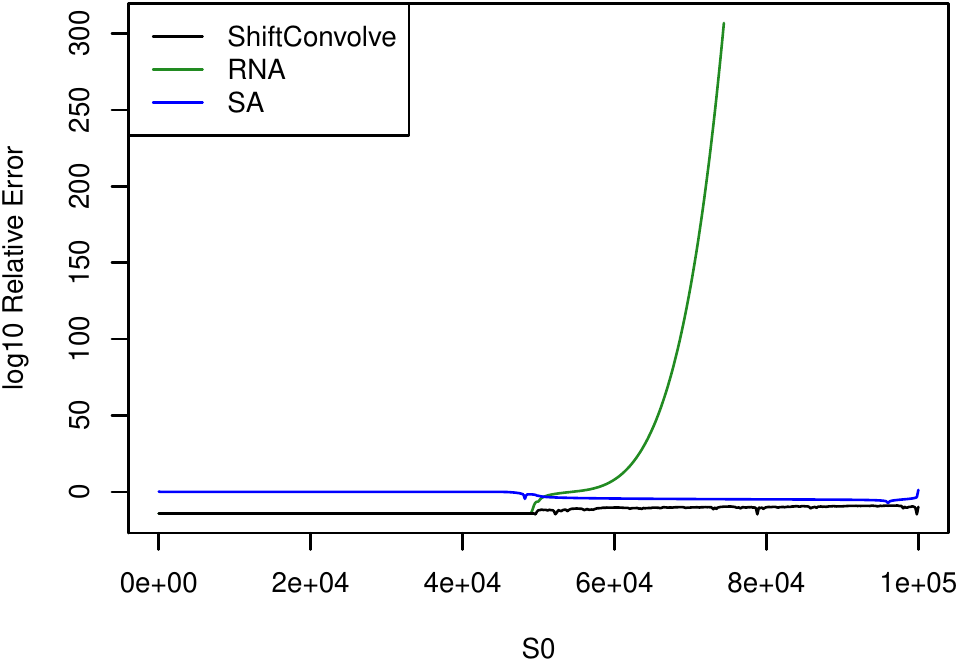}\tabularnewline
\end{tabular}\caption{\textbf{Inaccuracy of RNA.} The panels show examples where RNA fails
to accurately compute the right tail probability. (A) $N = 100$, $p_i$ sampled uniformly. The computed right tail probability
(p-value) reported by RNA compared with the accurate DC. (B) The corresponding
(log base 10 of) the relative error where we added for reference the
significantly smaller relative error of ShiftConvolve. (C) $N = 10^5$. $p_i$ sampled uniformly.
(D) The corresponding relative error. \label{fig:RNA-problem}}
\end{figure}

\newpage

\bibliographystyle{plain}
\bibliography{/Users/keich/References/journal_papers/stats/stat_papers,/Users/keich/References/journal_papers/papers}

\end{document}